\documentclass[11pt,draftcls,onecolumn]{IEEEtran}

\ifCLASSINFOpdf
\else
   \usepackage[dvips]{graphicx}
\fi
\usepackage{url}
\usepackage{graphicx}
\usepackage{amsmath}
\usepackage{comment}
\usepackage[utf8]{inputenc}
\usepackage[english]{babel}
\usepackage{color}
\usepackage{framed}
\usepackage{cite}
\usepackage{amsmath,amssymb,amsfonts}
\usepackage{graphicx}
\usepackage{textcomp}
\usepackage{cite}
\usepackage{bbold}
\usepackage{ifpdf}
\usepackage{times}
\usepackage{layout}
\usepackage{float}
\usepackage{afterpage}
\usepackage{amsmath}
\usepackage{amstext}
\usepackage{amssymb,bm}
\usepackage{latexsym}
\usepackage{color}
\usepackage{flexisym}
\usepackage{dblfloatfix}    
\usepackage{kantlipsum}
\usepackage{subcaption}
\usepackage{graphicx}
\usepackage{amsmath}
\usepackage{amsthm}
\usepackage{graphicx}
\usepackage{caption}
\usepackage{subcaption}
\usepackage{booktabs}
\usepackage{multicol}
\usepackage{lipsum}

\usepackage{enumitem}
\usepackage{amsmath}
\usepackage{algpseudocode}
\usepackage{amsthm}
\usepackage{dsfont}
\usepackage{algorithm}
\usepackage{thmtools}
\usepackage{amssymb}
\usepackage[dvipsnames]{xcolor}

\def\bI{\boldsymbol{I}}

\def\bP{\boldsymbol{P}}

\def\bR{\boldsymbol{R}}

\def\bmu{\boldsymbol{\mu}}

\def\bSigma{\boldsymbol{\Sigma}}
\def\btau{\boldsymbol{\tau}}

\newtheorem{theorem}{Theorem}

\newtheorem{lemma}{Lemma}

\theoremstyle{definition}

\algnewcommand\algorithmicinput{\textbf{Input:}}
\algnewcommand\Input{\item[\algorithmicinput]}
\algnewcommand\algorithmicoutput{\textbf{Output:}}
\algnewcommand\Output{\item[\algorithmicoutput]}
\algnewcommand\algorithmicinit{\textbf{Initialize:}}
\algnewcommand\Init{\item[\algorithmicinit]}

\makeatletter
\newcommand*{\rom}[1]{\expandafter\@slowromancap\romannumeral #1@}
\makeatother

\begin{document}

\title{Covariance Recovery for One-Bit Sampled  Data
With Time-Varying Sampling  Thresholds---\\ Part~\rom{1}: Stationary Signals}

\author{Arian Eamaz, \IEEEmembership{Student Member, IEEE}, Farhang Yeganegi,  and \\ Mojtaba Soltanalian, \IEEEmembership{Senior Member, IEEE}
\thanks{This work was supported in part by the National Science Foundation Grants CCF-1704401 and ECCS-1809225. Parts of this work were presented at the International Conference on Acoustics, Speech and Signal Processing (ICASSP) 2021, held in Toronto, Canada \cite{eamaz2021modified}. The first two authors contributed equally to this work.}
\thanks{A. Eamaz, F. Yeganegi and M. Soltanalian are with the Department of Electrical and Computer Engineering, University of Illinois Chicago, Chicago, IL 60607, USA (e-mails: \emph{\{aeamaz2, fyegan2, msol\}@uic.edu}).}
}

\markboth{IEEE TRANSACTIONS ON SIGNAL PROCESSING, 2022}
{Shell \MakeLowercase{\textit{et al.}}: Bare Demo of IEEEtran.cls for IEEE Journals}
\maketitle

\begin{abstract}
One-bit quantization, which relies on comparing the signals of interest with given threshold levels, has attracted considerable attention in signal processing for communications and sensing. A useful tool for covariance recovery in such settings is the \emph{arcsine law}, that estimates the normalized covariance matrix of zero-mean stationary input signals. This relation, however, only considers a zero sampling threshold, which can cause a remarkable information loss. In this paper, the idea of the arcsine law is extended to the case where one-bit analog-to-digital converters (ADCs) apply time-varying thresholds. Specifically, three distinct approaches are proposed, investigated, and compared, to recover the autocorrelation sequence of the stationary signals of interest. Additionally, we will study a modification of the \emph{Bussgang law}, a famous relation facilitating the recovery of the cross-correlation between the one-bit sampled data and the zero-mean stationary input signal. Similar to the case of the arcsine law, the Bussgang law only considers a zero sampling threshold. This relation is also extended to accommodate the more general case of time-varying thresholds for the stationary input signals.
\end{abstract}


\IEEEpeerreviewmaketitle

\section{Introduction}
\IEEEPARstart{D}{igital} signal processing typically requires the quantization of the signals of interest through analog-to-digital converters (ADCs). In high resolution settings, a very large number of quantization levels is required in order to represent the original continuous signal. However, this leads to some difficulties in modern applications where the signals of interest have large bandwidths, and may pass through several RF chains that require using a plethora of ADCs. Moreover, the overall power consumption and manufacturing cost of ADCs, and chip area grows exponentially with the number of quantization bits. Such drawbacks lend support to the idea of utilizing fewer bits for sampling. The most extreme version of this idea would be to use \emph{one-bit quantization}, in which ADCs are merely comparing the signals with given threshold levels, producing sign ($\pm1$) outputs. This allows for sampling at a very high rate, with a significantly lower cost and energy consumption compared to conventional ADCs \cite{instrumentsanalog,mezghani2018blind,ameri2018one,sedighi2020one}.

In the context of one-bit sampling, until recently, most researchers approached the problem of estimating signal parameters by comparing the signal with a fixed threshold, usually zero. This introduces difficulties in the recovery of the signal amplitude. On the other hand, recent works have shown enhanced estimation performance for the signal parameters by employing time-varying thresholds \cite{qian2017admm,gianelli2016one,khobahi2020model,khobahi2018signal,wang2017angular,xi2020gridless}.\par
The arcsine law is a fundamental statistical property of one-bit sampling \cite{van1966spectrum,jacovitti1994estimation,jacovitti1987methods,bar2002doa}, which connects the covariance of an unquantized signal with that of its quantized counterpart \cite{liu2017one,ameri2018one}. An important disadvantage of the arcsine law is, however, that the one-bit quantization threshold is considered to be zero, which leads to a considerable loss of information. 
In this paper, we present a new approach to extending the arcsine law in the context of time-varying sampling thresholds, which can recover the covariance values of the input unquantized signal with accuracy. In particular, we further expand on the ideas we presented in \cite{eamaz2021modified} by employing several competing recovery approaches. Moreover, we propose a new formalism for the Bussgang law \cite{Bussgang1952crosscorrelation,ameri2018one} in the context of time-varying thresholds, which is referred to as the \emph{modified Bussgang law}.

\subsection{Contributions of the Paper}
In this paper, we will study the covariance recovery in one-bit quantization systems with time-varying thresholds. We present a theorem demonstrating an integral representation of the relation between the autocorrelation function of the one-bit sampled data and the covariance matrix elements of the input signal. The obtained integral appears to be difficult to evaluate analytically. To approach this problem, we employ a one-point piece-wise \text{Padé} approximation (PA) to recast the integrands as rational expressions which are readily integrable. Next, we formulate an estimation criterion to recover the desired parameters which are the input signal variance and the input autocorrelation values. In the next step, we present the idea of fast input covariance matrix recovery based on a lemma which relates the mean of the one-bit sampled data to the input signal variance. Furthermore, alongside the PA technique, two well-known numerical integration approaches are employed to recover the input autocorrelation values using the proposed fast recovery algorithm; i.e. the Gauss-Legendre quadrature and the Monte-Carlo integration techniques. Lastly, the modified Bussgang law is presented considering time-varying thresholds. By using the modified Bussgang law, the cross-correlation matrix elements between the input signal and the one-bit sampled signal can be recovered. Numerical examples are presented to showcase the effectiveness of the proposed ideas and to provide an avenue for their comparison.

\subsection{Organization of the Paper}
Section~\rom{2} is dedicated to obtaining the autocorrelation function of the one-bit sampled data with time-varying thresholds in the case of stationary inputs. In Section~\rom{3}, the \text{Padé} Approximation (PA) is utilized to recover the input signal autocorrelation sequence. Moreover, a useful lemma is presented which relates the mean of the one-bit sampled data to the input signal variance laying the ground for a fast input covariance matrix recovery. Sections~\rom{4} and \rom{5} will present two famous numerical integration techniques applied to our fast covariance recovery problem; namely, the Gauss-Legendre quadrature and the Monte-Carlo integration methods. Section~\rom{6} is where the various methods proposed for covariance recovery are compared. The modified Bussgang law for time-varying thresholds in the case of stationary signals is obtained in Section~\rom{7}. Finally, Section~\rom{8} concludes the paper.
\vspace{5pt}

\underline{\emph{Notation:}}
We use bold lowercase letters for vectors, bold uppercase letters for matrices, and uppercase letters for matrix elements. For instance, $\bR_{\mathbf{x}}$ and $R_{\mathbf{x}}(i,j)$ denote the autocorrelation matrix and the $ij$-th element of the autocorrelation matrix of the vector $\mathbf{x}$, respectively. $(\cdot)^{\top}$ and $(\cdot)^{\mathrm{H}}$ denote the vector/matrix transpose, and the Hermitian transpose, respectively. $[a_{ij}]^{N_{1}\times N_{2}}$ is a $N_{1}\times N_{2}$ matrix with $a_{ij}$ as its $ij$-th element. $\mathbb{E}\left\{.\right\}$ denotes the expected value operator. The $Q$-function is defined as
\begin{equation}
\label{eq:135}
\begin{aligned}
Q(x) &= \frac{1}{\sqrt{2 \pi}} \int_{x}^{\infty} e^{-\frac{z^{2}}{2}} \,dz,\\
Q(x) &= 1-Q(-x) = \frac{1}{2}-\frac{1}{2}\operatorname{erf}\left(\frac{x}{\sqrt{2}}\right).
\end{aligned}
\end{equation}
where $\operatorname{erf}(.)$ is the associated error function. Further, $Q^{-1}(x)$ is an inverse $Q$-function.
Finally, the incomplete gamma function is given by
\begin{equation}
\label{eq:136}
\Gamma(s, x)=\int_{x}^{\infty} z^{s-1} e^{-z} \,dz.
\end{equation}
\section{Modified Arcsine Law for Time-Varying Thresholds}
Consider a zero-mean stationary Gaussian input signal, $\mathbf{x}\sim\mathcal{N}\left(\mathbf{0},\bR_{\mathbf{x}}\right)$, where $\bR_{\mathbf{x}}$ is a Toeplitz matrix associated with the autocorrelation function of $\mathbf{x}$, denoted as $R_{\mathbf{x}}$. The input signal $\mathbf{x}\in \mathbb{R}^{N}$ is an arbitrary temporal sequence of the distribution ensembles $\left\{\mathbf{x}(k)\right\}$ where $k \in \left\{1,\cdots,N_{\mathbf{x}}\right\}$. Suppose $x_i$ and $x_j$ are the $i$th and $j$th entries of $\mathbf{x}$, and $\mathbf{y}=f(\mathbf{x})$ is the output of a process in which $f(x)$ is the sign function. The autocorrelation function of the output, denoted by $R_{\mathbf{y}}(l)$, with $l=|i-j|$, is connected to that of $\mathbf{x}$ via the arcsine law~\cite{van1966spectrum,jacovitti1994estimation,jacovitti1987methods}:
\begin{equation}
\label{eq:1}
R_{\mathbf{y}}(i,j)=R_{\mathbf{y}}(l)=\mathbb{E}\left\{y_{i}y_{j}\right\}=\frac{2}{\pi}\sin^{-1}\left(\frac{R_{\mathbf{x}}(l)}{R_{\mathbf{x}}(0)}\right),
\end{equation}
where $y_{i}$ and $y_{j}$ are the $i$th and $j$th entries of $\mathbf{y}$, and $R_{\mathbf{x}}(l)$ denotes the input signal autocorrelation for lag $l$.

\subsection{Autocorrelation Function of the One-Bit Sampled Signal With Time-Varying Thresholds}
We consider a non-zero time-varying Gaussian threshold~$\btau$ that is independent of the input signal, with the distribution $\btau\sim\mathcal{N}\left(\mathbf{d}=\mathbf{1}d,\bSigma\right)$. We define a new random process $\mathbf{w}$ such that $\mathbf{w}=\mathbf{x}-\btau$. Clearly, $\mathbf{w}$ is a Gaussian stochastic process with $\mathbf{w}\sim\mathcal{N}\left(-\mathbf{d},\bR_{\mathbf{x}}+\bSigma=\bP\right)$. The autocorrelation function of the one-bit quantized output process for lag $l$ is studied in the following.
\begin{theorem}
\label{theorem_1}
Suppose 
$p_{l}$ and $p_{0}$ denote the autocorrelation term for lag $l\geq 1$, and the variance of~$\mathbf{w}$, respectively. Consider the one-bit quantized random variable $\mathbf{y}=f(\mathbf{w})$. Then, the autocorrelation function of $\mathbf{y}$ takes the form
\begin{equation}
\label{eq:112}
\begin{aligned}
R_{\mathbf{y}}(l)=\frac{e^{\frac{-d^{2}}{p_{0}+p_{l}}}}{\pi\sqrt{\left(p_{0}^{2}-p_{l}^{2}\right)}}\left\{ \int_{0}^{\frac{\pi}{2}} \frac{1}{\beta_{s}}+\sqrt{\frac{\pi}{\beta_{s}}} \frac{\alpha_{s}}{2\beta_{s}} e^{\frac{\alpha_{s}^{2}}{4 \beta_{s}}}\right.\\\left.-\sqrt{\frac{\pi}{\beta_{s}}} \frac{\alpha_{s}}{\beta_{s}} Q\left(\frac{\alpha_{s}}{\sqrt{2 \beta_{s}}}\right) e^{\frac{\alpha_{s}^{2}}{4 \beta_{s}}} \,d\theta\right\}-1,
\end{aligned}
\end{equation}
where $\alpha_{s}$ and $\beta_{s}$ are evaluated as
\begin{equation}
\label{eq:113}
\begin{aligned}
\alpha_{s} &= \frac{d\left(\sin\theta+\cos\theta\right)}{p_0+p_{l}},\\ \beta_{s} &= \frac{p_0-p_{l}\sin 2\theta}{2(p_{0}^2-p_{l}^2)}.\\
\end{aligned}
\end{equation}
\end{theorem}
\begin{proof}
The covariance matrix of $\mathbf{y}$ can be written as
\begin{equation}
\label{eq:110}
\begin{aligned}
\bR_{\mathbf{y}} &= \mathbb{E}\left\{\mathbf{y}\mathbf{y}^{\mathrm{H}}\right\},\\
&= \frac{1}{\sqrt{(2\pi)^{N}|\bP|}}\int^{\infty}_{-\infty}\mathcal{I}(\mathbf{w}) e^{-\frac{1}{2}(\mathbf{w}+\mathbf{d})^{\mathrm{H}}\bP^{-1}(\mathbf{w}+\mathbf{d})} \mathbf{\,dw},
\end{aligned}
\end{equation}
where $\mathcal{I}(\mathbf{w})=\mathbf{f(w)}\mathbf{f(w)}^{\mathrm{H}}$ and $\mathbf{f(w)}=\left[f(w_{j})\right]^{N}_{j=1}$ is a column vector. Clearly, $\mathcal{I}$ is a matrix including only entries of the form $\pm 1$.
Note that one can write the output covariance matrix as \begin{equation}
\label{eq:111}
\bR_{\mathbf{y}}=\left[\mathbb{E}\{y_{i}y_{j}\}\right]^{N\times N}.
\end{equation}
Therefore, the autocorrelation of $f(w_i)$ and $f(w_j)$ is given by
\begin{equation}
\label{eq:71}
\begin{aligned}
R_{\mathbf{y}}(i,j)&=\mathbb{E}\{y_{i}y_{j}\},\\
&= \mathbb{E}\left\{f(w_{i})f(w_{j})\right\},\\
&= \int^{\infty}_{-\infty}\int^{\infty}_{-\infty} f(w_{i}) f(w_{j}) p(w_{i},w_{j}) \,dw_{i}\,dw_{j},
\end{aligned}
\end{equation}
where $p(w_{i},w_{j})$ is the joint Gaussian probability distribution, that can be obtained as 
\begin{eqnarray}
\label{eq:72}
p(w_{i},w_{j})= \qquad\qquad\qquad\qquad \qquad\qquad\qquad\qquad \\ \nonumber \qquad \frac{1}{2\pi\sqrt{p_{0}^2-p_{l}^2}} ~ e^{-\frac{(w_i+d)^2p_{0}+(w_j+d)^2p_{0}-2p_{l}(w_i+d)(w_j+d)}{2(p_{0}^2-p_{l}^2)}}.
\end{eqnarray}
By substituting (\ref{eq:72}) in (\ref{eq:71}), the output autocorrelation function $R_{\mathbf{y}}(i,j)$ can be evaluated as \cite{eamaz2021modified},
\begin{equation}
\label{eq:2}
R_{\mathbf{y}}(i,j)\hspace{-.1cm} = \frac{1}{2\pi\sqrt{p_{0}^2-p_{l}^2}}\int_{-\infty}^{\infty} \int_{-\infty}^{\infty}\hspace{-.1cm} f(w_i)f(w_j)e^{\lambda(d)} \,dw_i\,dw_j
\end{equation}
where $\lambda(d)$ is defined as follows: 
\begin{equation}
\label{eq:3}
\lambda(d)\hspace{-.1cm}=\frac{(w_i+d)^2p_{0}+(w_j+d)^2p_{0}-2p_{l}(w_i+d)(w_j+d)}{-2(p_{0}^2-p_{l}^2)}.
\end{equation}
The autocorrelation function in (\ref{eq:2}) can be rewritten as
\begin{equation}
\label{eq:4}
\begin{aligned}
R_{\mathbf{y}}(i,j)\hspace{-.1cm}=\,& \frac{1}{2\pi\sqrt{p_{0}^2-p_{l}^2}}\left(\int_{0}^{\infty} \int_{0}^{\infty}e^{\lambda(d)}\,dw_i\,dw_j\right.\\ &+\int_{-\infty}^{0}\int_{-\infty}^{0}e^{\lambda(d)}\,dw_i\,dw_j\\ &-\int_{0}^{\infty}\int_{-\infty}^{0}e^{\lambda(d)}\,dw_i\,dw_j\\ &\left.-\int_{-\infty}^{0}\int_{0}^{\infty}e^{\lambda(d)}\,dw_i\,dw_j\right).
\end{aligned}
\end{equation}
We can simplify (\ref{eq:4}) using the relation
\begin{equation}
\label{eq:5}
\frac{1}{2\pi\sqrt{p_{0}^2-p_{l}^2}}\int_{-\infty}^{\infty} \int_{-\infty}^{\infty} e^{\lambda(d)} \,dw_i\,dw_j=1.
\end{equation}
In fact, using \eqref{eq:5} one can verify that
\begin{equation}
\label{eq:6}
\begin{aligned}
R_{\mathbf{y}}(i,j) = \frac{1}{\pi\sqrt{p_{0}^2-p_{l}^2}}\int_{0}^{\infty}\hspace{-.1cm} \int_{0}^{\infty}& \hspace{-.1cm}\left(e^{\lambda(d)}+e^{\lambda(-d)}\right) \,dw_i\,dw_j\\&-1.
\end{aligned}
\end{equation}
By employing polar coordinates $w_i=\rho \cos \theta$, $w_j=\rho \sin \theta$, we can recast the integral in (\ref{eq:6}) as
\begin{equation}
\label{eq:7}
\begin{aligned}
R_{\mathbf{y}}(i,j) = \frac{e^{\frac{-d^2}{p_0+p_{l}}}}{\pi\sqrt{p_{0}^2-p_{l}^2}}\int_{0}^{\frac{\pi}{2}} \hspace{-.1cm} \int_{0}^{\infty}& \hspace{-.1cm} e^{-\beta\rho^2} \hspace{-.1cm}\left(e^{-\alpha\rho}+e^{\alpha\rho}\right)\rho\,d\rho\,d\theta\\ &-1,
\end{aligned}
\end{equation}
where
\begin{equation}
\label{eq:8}
\begin{aligned}
\alpha_{s} &= \frac{d\left(\sin\theta+\cos\theta\right)}{p_0+p_{l}},\\ \beta_{s} &= \frac{p_0-p_{l}\sin 2\theta}{2(p_{0}^2-p_{l}^2)}.
\end{aligned}
\end{equation}
Let $R_{\mathbf{y}}(l)=R_{\mathbf{y}}(i,j)$ with $l=|i-j|$.  Integrating (\ref{eq:7}) with respect to $\rho$ leads to 
\begin{equation}
\label{eq:9}
\begin{aligned}
R_{\mathbf{y}}(l)=\frac{e^{\frac{-d^{2}}{p_{0}+p_{l}}}}{\pi\sqrt{\left(p_{0}^{2}-p_{l}^{2}\right)}}\left\{ \int_{0}^{\frac{\pi}{2}} \frac{1}{\beta_{s}}+\sqrt{\frac{\pi}{\beta_{s}}} \frac{\alpha_{s}}{2\beta_{s}} e^{\frac{\alpha_{s}^{2}}{4 \beta_{s}}}\right.\\\left.-\sqrt{\frac{\pi}{\beta_{s}}} \frac{\alpha_{s}}{\beta_{s}} Q\left(\frac{\alpha_{s}}{\sqrt{2 \beta_{s}}}\right) e^{\frac{\alpha_{s}^{2}}{4 \beta_{s}}} \,d \theta\right\}-1,
\end{aligned}
\end{equation}
a transition for which you can find more detailed derivations in Appendix~A. This completes the proof.
\end{proof}

It remains to evaluate the integral in (\ref{eq:112}) in terms of $p_0$ and $\{p_{l}\}$,  which have to be estimated---a task that is central to our efforts in the rest of this paper. Finding $p_0$ and $\{p_{l}\}$ results in input variance and autocorrelation recovery, which can be achieved by considering the relation:
\begin{equation}
\label{eq:10}
\bR_{\mathbf{x}}(i,j) = \bP(i,j)-\bSigma(i,j).
\end{equation}
For $i=j$, the input variance is thus given by
$\bR_{\mathbf{x}}(i,i) = r_{0} = p_0-\bSigma(i,i),$ 
while for $i\neq j$, we have the input autocorrelation for lag $l=|i-j|$ as $\bR_{\mathbf{x}}(i,j) = \bR_{\mathbf{x}}(l) = r_{l} = p_{l}-\bSigma(i,j)$.

Note that evaluating the integral in (\ref{eq:112}) does not appear to be amenable to a closed-form solution in its general form. Therefore, in the following, we resort to various approximations to facilitate its evaluation, leading to the recovery of the input signal covariance matrix elements.

\section{Analytic Approach for Covariance Recovery}
\label{Proposed method}
To enable an approximation of the autocorrelation values in (\ref{eq:112}), we first resort to the rational \text{Padé} approximation (PA) \cite{basdevant1972pade,brezinski1995taste,gonnet2013robust}. This lays the ground for the recovery of $p_{0}$ and $\{p_{l}\}$, as discussed in Section \ref{subsec:2}.

\subsection{Proposed Rational Approximation}
\label{subsec:1}
According to \cite{chiani2003new}, the $Q$-function is well-approximated by the $\bar{Q}$-function as,
\begin{equation}
\label{eq:11}
\bar{Q}\left(x\right) = \frac{1}{12} e^{\frac{-x^2}{2}} +\frac{1}{4} e^{\frac{-2x^2}{3}}, \quad x > 0.
\end{equation}
We further note that the integral in (\ref{eq:112}) may be evaluated by substituting $D_{1}\left(\theta;p_{0},p_{l},d\right)=\sqrt{\frac{\pi}{\beta_{s}}} \frac{\alpha_{s}}{\beta_{s}} Q\left(\frac{\alpha_{s}}{\sqrt{2 \beta_{s}}}\right) e^{\frac{\alpha_{s}^{2}}{4 \beta_{s}}}$ and $D_{2}\left(\theta;p_{0},p_{l},d\right)\  =\sqrt{\frac{\pi}{\beta_{s}}} \frac{\alpha_{s}}{2\beta_{s}} e^{\frac{\alpha_{s}^{2}}{4 \beta_{s}}}$ with \text{Padé} approximants, that yield the best approximation of a function by a rational function of given order through the \emph{moment matching} technique.

For the sake of completeness, herein we present a brief introduction of the PA method. Suppose $I(t)$ is an \emph{analytic function} at point $t=0$ with the Taylor series:
\begin{equation}
\label{eq:12}
I(t) = \sum_{n=0}^{\infty} c_{n}t^{n}, \quad c_{n}\in \mathbb{R} .
\end{equation}
The PA of order $\left[L/M\right]$ for $I(t)$, denoted by $P^{\left[L/M\right]}(t)$, is defined as a rational function in the form \cite{brezinski1995taste,gonnet2013robust}:
\begin{equation}
\label{eq:13}
P^{\left[L/M\right]}(t) \triangleq \frac{\sum_{n=0}^{L}a_{n}t^{n}}{\sum_{n=0}^{M}b_{n}t^{n}}
\end{equation}
where the coefficients $\left\{a_{n}\right\}$ and $\left\{b_{n}\right\}$ are defined so that
\begin{equation}
\label{eq:14}
\begin{aligned}
\lim_{t\rightarrow 0}\quad \frac{\sum_{n=0}^{L}a_{n}t^{n}}{\sum_{n=0}^{M}b_{n}t^{n}}&=\sum_{n=0}^{L+M}c_{n}t^{n} +O(t^{L+M+1})
\end{aligned}
\end{equation}
with $b_{0}=1$. The moment matching technique is a method widely used to obtain the coefficients of PA. The coefficients $\left\{b_{n}\right\}$ are obtained through the linear system of equations~\cite{brezinski1995taste,gonnet2013robust}:
\begin{equation}
\label{eq:15}
\begin{aligned}
&\left[\begin{array}{cccc}
c_{L-M+1} & c_{L+M+2} & \cdots & c_{L} \\
\vdots & \vdots & \vdots & \vdots \\
c_{L-M+k} & c_{L-M+k+1} & \cdots & c_{L+k-1} \\
\vdots & \vdots & \vdots & \vdots \\
c_{L} & c_{L+1} & \cdots & c_{L+M-1}
\end{array}\right]\left[\begin{array}{c}
b_{M} \\
\vdots \\
b_{k} \\
\vdots \\
b_{1}
\end{array}\right]\\
&=-\left[\begin{array}{c}
c_{L+1} ~
\cdots ~
c_{L+k+1} ~
\cdots ~
c_{L+M}
\end{array}\right]^T
\end{aligned}
\end{equation}
where the matrix in the  left-hand side of (\ref{eq:15}) is a Hankel matrix. Clearly, the determinant of the Hankle matrix must be non-zero to permit a unique solution to the linear system.
The coefficients $\left\{a_{n}\right\}$ are obtained by backsubstitution \cite{eamaz2021modified,brezinski1995taste,gonnet2013robust}:
\begin{equation}
\label{eq:16}
a_{j}=c_{j}+\sum_{i=1}^{\min (M, j)} b_{i} c_{j-i}, \quad 0 \leq j \leq L.
\end{equation}
The selection of the PA order is an important task in approximation; see \cite{basdevant1972pade,brezinski1995taste,gonnet2013robust} for a related study. Note that the integration in (\ref{eq:112}) occurs in the interval $\theta \in \left[0,\frac{\pi}{2}\right]$. To have a better fitness, we use the idea of piece-wise PA. Owing to the fact that the functions $D_{1}\left(\theta;p_{0},p_{l},d\right)$ and $D_{2}\left(\theta;p_{0},p_{l},d\right)$ have their extremum at $\theta=\frac{\pi}{4}$, the selection of three distinct intervals $\left[0,\frac{\pi}{8}\right]$, $\left[\frac{\pi}{8},\frac{3\pi}{8}\right]$, and $\left[\frac{3\pi}{8},\frac{\pi}{2}\right]$ with the expansion points $\theta \in \left\{0,\frac{\pi}{4},\frac{\pi}{2}\right\}$ paves the way for a convenient approximation, with extra boundary points $\frac{\pi}{8}$ and $\frac{3\pi}{8}$ making the chosen intervals symmetric. Moreover, choosing more expansion points to approximate our integrands in (\ref{eq:112}) is not appropriate due to its relatively large computational burden which is caused by relatively large approximated coefficients. By adopting the above piece-wise scheme, the function $D_{2}\left(\theta;p_{0},p_{l},d\right)$ can be approximated as,
\begin{equation}
\label{eq:17}
\begin{aligned}
\theta \in \left[0,\frac{\pi}{8}\right] \cup \left[\frac{3\pi}{8},\frac{\pi}{2}\right]&: \sqrt{\frac{\pi}{\beta_{s}}} \frac{\alpha_{s}}{2\beta_{s}} e^{\frac{\alpha_{s}^{2}}{4 \beta_{s}}}\approx \frac{e+s\theta}{k+g\theta+h\theta^2},\\
\theta \in \left[\frac{\pi}{8},\frac{3\pi}{8}\right]&: \sqrt{\frac{\pi}{\beta_{s}}} \frac{\alpha_{s}}{2\beta_{s}} e^{\frac{\alpha_{s}^{2}}{4 \beta_{s}}}\approx \frac{z+u\theta+v\theta^2}{k\textprime+g\textprime\theta+h\textprime\theta^{2}}.
\end{aligned}
\end{equation}
A similar approximation with same orders can be proposed for $D_{1}\left(\theta;p_{0},p_{l},d\right)$. As mentioned earlier, the two functions $D_{1}\left(\theta;p_{0},p_{l},d\right)$ and  $D_{2}\left(\theta;p_{0},p_{l},d\right)$ should be analytic at the expansion points (which can be easily verified in this case). Accordingly, many diagonal and subdiagonal elements of PA with higher orders could be used; however, the aforementioned interval partitions appear to provide a good approximation while maintaining the simplicity of the integrands.

The first part of the integration in (\ref{eq:112}) can be analytically evaluated as
\begin{equation}
\label{eq:18}
\begin{aligned}
\int_{0}^{\frac{\pi}{2}} \frac{1}{\beta_{s}} \,d\theta = \sqrt{p_{0}^2-p_{l}^2}\left(\pi+2\tan^{-1}\left[\frac{p_{l}}{\sqrt{p_{0}^2-p_{l}^2}}\right]\right).
\end{aligned}
\end{equation}
Substituting $D_{2}\left(\theta;p_{0},p_{l},d\right)$ with its approximation and evaluating the integration in the associated parts of (\ref{eq:112}) results in:
\begin{equation}
\label{eq:19}
\begin{aligned}
\int_{0}^{\frac{\pi}{8}}& \sqrt{\frac{\pi}{\beta_{s}}}\frac{\alpha_{s}}{2\beta_{s}}e^{\frac{\alpha_{s}^2}{4\beta_{s}}} \,d\theta \approx \frac{s}{2h}\ln{\left(\frac{\left|k+\frac{\pi g}{8}+\frac{\pi^2h}{64}\right|}{\left|k\right|}\right)}+\\&\frac{2eh-sg}{h\sqrt{4hk-g^2}}\tan^{-1}\left(\frac{\pi h\sqrt{4hk-g^2}}{16hk+\pi gh}\right),
\end{aligned}
\end{equation}
\begin{equation}
\label{eq:20}
\begin{aligned}
\int_{\frac{\pi}{8}}^{\frac{3\pi}{8}}& \sqrt{\frac{\pi}{\beta_{s}}}\frac{\alpha_{s}}{2\beta_{s}}e^{\frac{\alpha_{s}^2}{4\beta_{s}}} \,d\theta \approx \frac{\pi v}{4h\textprime}+ \\
&\frac{uh\textprime-vg\textprime}{2h\textprime^2}\ln\left(\frac{\left|64k\textprime+9\pi^2h\textprime+24\pi g\textprime\right|}{\left|64k\textprime+\pi^2h\textprime+8\pi g\textprime\right|}\right)+ \\
&\frac{2vh\textprime k\textprime-2zh\textprime^2+ug\textprime h\textprime-vg\textprime^2}{h\textprime^2\sqrt{4k\textprime h\textprime-g\textprime^2}}\\
&\hspace{1.7cm}\tan^{-1}\left(\frac{-8\pi h\textprime\sqrt{4h\textprime k\textprime-g\textprime^2}}{64h\textprime k\textprime +3\pi^2h\textprime^2+16\pi h\textprime g\textprime}\right),
\end{aligned}
\end{equation}
\begin{equation}
\label{eq:21}
\begin{aligned}
\int_{\frac{3\pi}{8}}^{\frac{\pi}{2}}& \sqrt{\frac{\pi}{\beta_{s}}}\frac{\alpha_{s}}{2\beta_{s}}e^{\frac{\alpha_{s}^2}{4\beta_{s}}} \,d\theta \approx \frac{s}{2h} \ln\left(\frac{\left|k+\frac{\pi g}{2}+\frac{\pi^2h}{4}\right|}{\left|k+\frac{3\pi g}{8}+\frac{9\pi^2 h}{64}\right|}\right)+\\&\frac{2eh-sg}{h\sqrt{4kh-g^2}}\tan^{-1}\left(\frac{\pi h\sqrt{4hk-g^2}}{16kh+3\pi^2h^2+7\pi hg}\right).
\end{aligned}
\end{equation}
Similar approximations can be obtained for terms associated with the function $D_{1}\left(\theta;p_{0},p_{l},d\right)$.

\subsection{Recovery Criterion}
\label{subsec:2}
In this subsection, $p_0$ and $\{p_{l}\}$  are estimated by formulating a minimization problem. For this purpose, one may consider the following criterion \cite{eamaz2021modified}:
\begin{equation}
\label{eq:22}
\begin{aligned}
&\bar{C}(p_0,p_{l}) \triangleq \log\left(\left|R_{\mathbf{y}}(l)-\frac{e^{\frac{-d^{2}}{p_{0}+p_{l}}}}{\pi\sqrt{\left(p_{0}^{2}-p_{l}^{2}\right)}}\left\{ \int_{0}^{\frac{\pi}{2}} \frac{1}{\beta_{s}}\right.\right.\right.\\& \hspace{-.3cm} \left.\left.\left.+\sqrt{\frac{\pi}{\beta_{s}}} \frac{\alpha_{s}}{2\beta_{s}} e^{\frac{\alpha_{s}^{2}}{4 \beta_{s}}}-\sqrt{\frac{\pi}{\beta_{s}}} \frac{\alpha_{s}}{\beta_{s}} Q\left(\frac{\alpha_{s}}{\sqrt{2 \beta_{s}}}\right) e^{\frac{\alpha_{s}^{2}}{4 \beta_{s}}} d \theta\right\}+1\right|^2\right),
\end{aligned}
\end{equation}
where the autocorrelation of output signal ($R_{\mathbf{y}}$) can be estimated with the given sign vector ($\mathbf{y}$) via the sample covariance matrix \cite{hayes2009statistical},
\begin{equation}
\label{eq:23}
\begin{aligned}
\bR_{\mathbf{y}}\approx \frac{1}{N_{\mathbf{x}}} \sum_{k=1}^{N_{\mathbf{x}}} \mathbf{y}(k) \mathbf{y}(k)^{\mathrm{H}}.
\end{aligned}
\end{equation}
Note that by now we have derived an approximated version of~(\ref{eq:112}). Let $H_{s}(p_0,p_l)$ denote this approximation. Therefore, we can alternatively consider the criterion:
\begin{equation}
\label{eq:24}
\begin{aligned}
C(p_0,p_{l}) &\triangleq \log\left(\left|R_{\mathbf{y}}(l)-H_{s}(p_0,p_l)\right|^2\right).
\end{aligned}
\end{equation}
A numerical investigation of \eqref{eq:24} reveals that it is highly multi-modal, with many local minima---see Fig.~\ref{figure_1} for an example of the optimization landscape of $C(p_0,p_{l})$. To filter out the undesired local minima, we resort to  constraints re-enforcing the behaviour of an autocorrelation function. More precisely, we will consider the minimization problem:
\begin{equation}
\label{eq:25}
\begin{aligned}
\mathcal{P}_{\ell}&: &\min_{p_0,p_{l}}& &C(p_0,p_{l}),& &\mbox{s.t.}& &p_0^2 \geq p_{l}^2,& & p_0 \geq 0,
\end{aligned}
\end{equation}
where the first inequality constraint in (\ref{eq:25}) is imposed to ensure that the magnitude of the diagonal elements of the covariance matrix of $\mathbf{w}$ is greater than the magnitude of the off-diagonal elements. The non-convex problem in (\ref{eq:25}) may be solved via the gradient descent numerical optimization approach by employing multiple random initial points.
Once $p_0$ and $\{p_{l}\}$ are obtained, one can estimate the autocorrelation values of $\mathbf{x}$ via \eqref{eq:10}.

\begin{figure}[tb]
	\center{\includegraphics[width=0.6\textwidth]
		{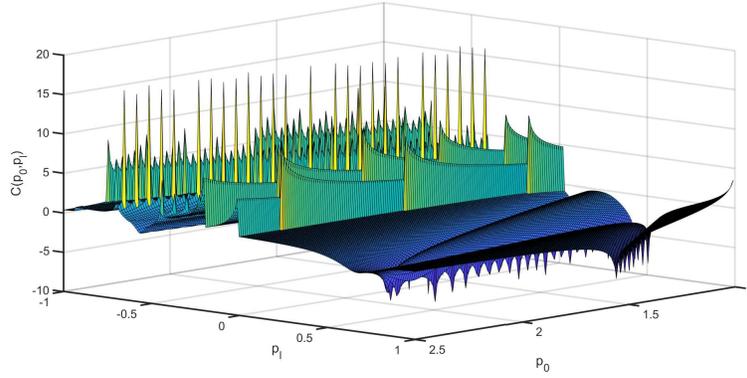}}
	\caption{Example plot of the estimation criterion $C(p_0,p_{l})$ with respect to $p_0$ and $p_{l}$ showing its multi-modality, i.e. having multiple local optima.}
	\label{figure_1}
\end{figure}

\subsection{Optimal Variance Estimation for Fast Covariance Recovery}
\label{sec:fast_recovery}
In Section~\ref{subsec:2}, we suggested that the unknown variables $p_{0}$ and $\{p_{l}\}$ may be recovered through the non-convex program (\ref{eq:25}) with a well-chosen initial point. Nevertheless, solving such a multivariate non-convex problem can costly and finding a proper initial point can be challenging. In this section, we discuss how one can accelerate finding the optimal point in the cost function (\ref{eq:24}). Namely, we introduce the idea of fast covariance matrix recovery by reducing the number of optimization variables. To make this happen, one can estimate the optimal variance  $p_{0}$ based on the following lemma:
\begin{lemma}
\label{remark_1}
The first moment (mean) of the one-bit sampled data, typically approximated as $\bmu\approx\frac{1}{N_{\mathbf{x}}} \sum^{N_{\mathbf{x}}}_{k=1}\mathbf{y}(k)$, depends on the threshold distribution and the power of sampled data via the relation:
\begin{equation}
\label{rem1} \bmu=\mathbb{E}\left\{\mathbf{y}\right\} =\mathbf{1} \mu= \mathbf{1}\left(2Q\left(\frac{d}{\sqrt{p_{0}}}\right)-1\right),
\end{equation}
\begin{proof}
We have
\begin{equation}
\label{pr1-1}
\mathbb{E}\left\{y_{i}\right\} = \int_{-\infty}^{+\infty} f(w_{i}) p(w_{i}) \,dw_{i},
\end{equation}
for $i\in \{1,\cdots,N\}$, where $p(w_{i})=\left(\sqrt{2\pi p_{0}}\right)^{-1} e^{\frac{-\left(w_{i}+d\right)^{2}}{2p_{0}}}$. We can further simplify (\ref{pr1-1}) as
\begin{equation}
\label{pr1-2}
\begin{aligned}
\mathbb{E}\left\{y_{i}\right\} &= -\int_{-\infty}^{0}p(w_{i}) \,dw_{i} + \int_{0}^{\infty} p(w_{i}) \,dw_{i}\\ &= 2\int_{0}^{+\infty} p(w_{i}) \,dw_{i}-1\\ &= 2Q\left(\frac{d}{\sqrt{p_{0}}}\right)-1
\end{aligned}
\end{equation}
which completes the proof.
\end{proof}
\end{lemma}
We observe that Lemma~\ref{remark_1} unveils a relationship between the input variance and the mean of one-bit sampled data. Therefore, in addition to (\ref{eq:112}), we have another equation to evaluate the variance $p_{0}$.
The input variance as evaluated via Lemma~\ref{remark_1} is given as
\begin{equation}
\label{eq:fast_1}
p_{0}^{\star} = \left(\frac{d}{Q^{-1}\left(\frac{\mu+1}{2}\right)}\right)^{2},
\end{equation}
where $p_{0}^{\star}$ denotes the optimal value of $p_{0}$. Moreover, according to Lemma~\ref{remark_1}, all elements of the one-bit sampled data mean are equal. However, because of using the  approximated mean, some elements can have a negligible difference with each other. In order to compensate these differences, an average of elements may be deployed. Subsequently, based on (\ref{eq:10}), the input variance can be obtained using $p_{0}^{\star}$.
Once $p_{0}$ is obtained, one can estimate $p_{l}$ based on (\ref{eq:112}). As a result, in the PA-based covariance recovery,  problem (\ref{eq:25}) boils down to the single-variable optimization problem,
\begin{equation}
\label{eq:fast_2}
\begin{aligned}
\mathcal{P}_{\ell}&: &\min_{p_{l}}& &C(p_{l}),& &\mbox{s.t.}& &-p_{0}^{\star} \leq p_{l} \leq p_{0}^{\star},
\end{aligned}
\end{equation}
where $C(p_{l})=C\left(p^{\star}_{0},p_{l}\right)$. The objective function of the above optimization problem is still multi-modal---see Fig.~\ref{fig_fast_2} for an example of the optimization landscape of $C(p_{l})$. 
However, the process of finding the optimal point has been made faster by choosing an one-dimensional slice $\left(p^{\star}_{0},p_{l}\right)$ of the feasible space of the objective function (\ref{eq:24}) containing the optimal value of the autocorrelation value $p_{l}$. Consequently, the dependency of the recovery algorithm to choosing an appropriate initial value for $p_{0}$ is eliminated. In other words, by optimal variance substitution in the objective function (\ref{eq:24}), we are effectively removing many poor local optima.

Similar to Section~\ref{subsec:2}, the non-convex problem in (\ref{eq:fast_2}) may be solved via the gradient descent numerical optimization approach by employing multiple random initial points.

\begin{figure}[t]
	\center{\includegraphics[width=0.6\textwidth]{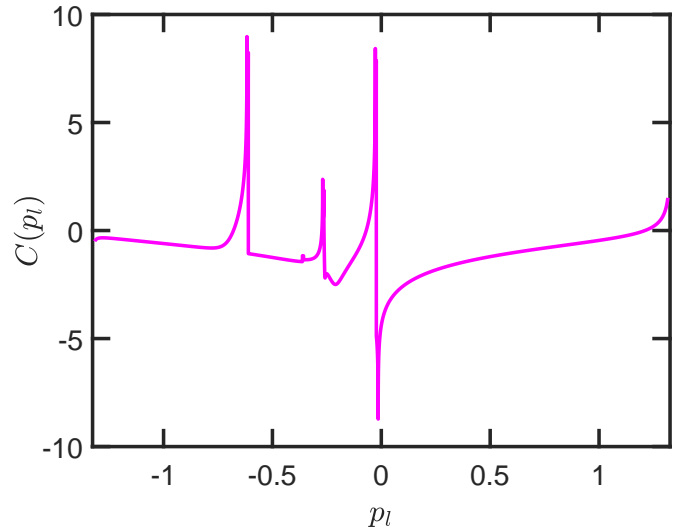}}
	\caption{Example plot of the simplified estimation criterion $C(p_{l})$ with respect to $p_{l}$ showing its multi-modality, i.e. having multiple local optima.}
	\label{fig_fast_2}
\end{figure}
\begin{figure}[t]
	\center{\includegraphics[width=0.6\textwidth]
		{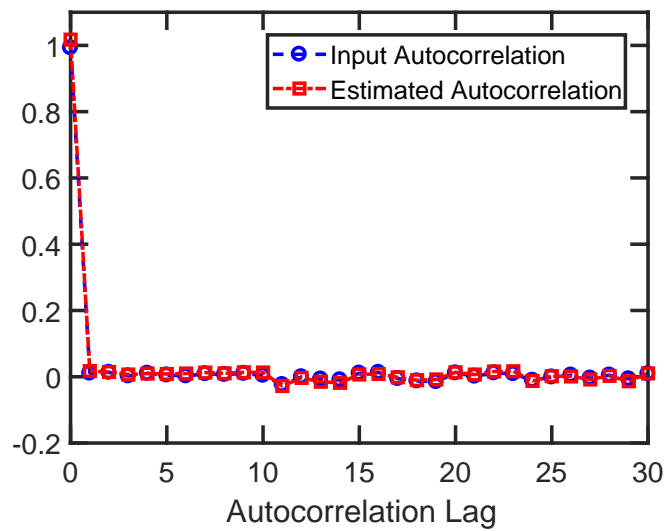}}
	\caption{Recovery of the input signal autocorrelation for a sequence of length $31$ from one-bit sampled data, with the true values plotted along the estimates.}
	\label{figure_2}
\end{figure}
\subsection{Numerical Results}
\label{sec:stationary_analytic}
In this section, we will examine the proposed method by comparing its recovery results with the true input signal autocorrelation values. In all experiments, the input signals were generated as zero-mean Gaussian sequences with unit variance. The number of states $N$ is $100$ ($\mathbf{x}\in \mathbb{R}^{100}$). Accordingly, we made use of the time-varying thresholds with $d=0.3$ and diagonal $\bSigma$ whose diagonal entries are equal to $0.4$. Note that the values of $d$ and $\bSigma$ are best chosen based on the application, considering the magnitude of the input signal.

To show the effectiveness of the proposed approach, we present an example of autocorrelation sequence recovery. The true input signal autocorrelation and the estimated autocorrelation values by our approach are shown in Fig.~\ref{figure_2} for a random sequence of length $31$. Fig.~\ref{figure_2} appears to confirm the possibility of recovering the autocorrelation values from one-bit sampled data with time-varying thresholds.

Next, we investigate the impact of a growing sample size in autocorrelation recovery, and in particular, the variance. We define the normalized mean square error (NMSE) of an estimate $\hat r_{0}$ of a variance $r_{0}$ as
\begin{equation}
\label{eq:26}
\begin{aligned} 
\mathrm{NMSE} \triangleq \frac{|r_{0}-{\hat r_{0}}|^{2}}{|r_{0}|^{2}}.
\end{aligned}
\end{equation}
Each data point presented is averaged over $15$ experiments. As can be seen in Fig.~\ref{figure_3}, the proposed method can estimate the variance elements of an input signal accurately. The results are obtained for the number of ensembles $N_{\mathbf{x}}\in \left\{1000, 3000, 6000,  10000\right\}$, with fixed $d$ and $\bSigma$ for each experiment. As expected, the accuracy of variance recovery will significantly enhance as the number of one-bit samples grows large.

To examine the efficacy of fast covariance matrix recovery method discussed in Section~\ref{sec:fast_recovery}, we consider the same setting for the input signal. Fig.~\ref{figure_3} shows the performance of (\ref{eq:fast_1}) in estimating the input variance. Each data point presented is averaged over $15$ experiments, in which we made use of time-varying thresholds with $d=0.7$ and $\bSigma=0.3\bI$, where $\bI$ denotes the identity matrix. Additionally, Fig.~\ref{fig_fast_3} confirms the possibility of input autocorrelation sequence recovery using (\ref{eq:fast_2}) when the parameters of the time-varying thresholds are set to $d=0.3$ and $\bSigma=0.4\bI$. Fig.~\ref{figure_3} reaffirms that by estimating the optimal variance from (\ref{eq:fast_2}), the accuracy of the variance recovery is improved. Interestingly, in our numerical experiments, solving the criterion (\ref{eq:25}) took $25$ times more CPU time than the single-variable problem proposed in (\ref{eq:fast_2}).

\section{Gaussian Quadrature Technique for Covariance Recovery}
In this section, we will adopt the Gauss-Legendre quadrature method, a well-known numerical integration technique, to evaluate the integral in (\ref{eq:112}). This lays the ground for the recovery of $\left\{p_{l}\right\}$ since $p_{0}$ is obtained by (\ref{eq:fast_1}). Finally, the efficacy of this approach in estimating the input autocorrelation values is numerically evaluated. We will present a brief review of the Gauss-Legendre quadrature technique in \ref{subsec:5}. We will then proceed to obtain an approximated version of (\ref{eq:112}) based on Gauss-Legendre quadrature rule to recover $\{p_{l}\}$, and subsequently, the input autocorrelation values in Section \ref{subsec:6}.

\begin{figure}[t]
	\center{\includegraphics[width=0.6\textwidth]{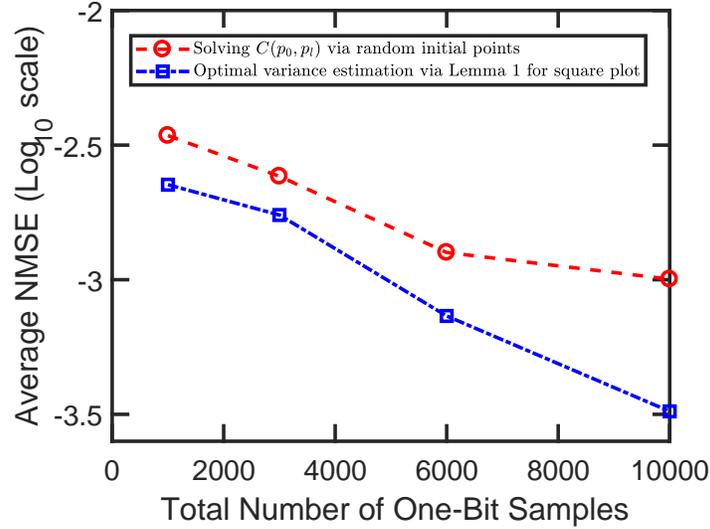}}
	\caption{Average NMSE for signal variance recovery for different one-bit sample sizes when (i) the non-convex program (\ref{eq:25}) with random initial points, and (ii) the closed form formula in Lemma~\ref{remark_1}, are used to evaluate the input variance.}
	\label{figure_3}
\end{figure}

\subsection{The Gauss-Legendre Quadrature Approach: A Short Introduction}
\label{subsec:5}
The quadrature rule is a famous approximation approach in the numerical analysis utilized to approximate the definite integral of a function, which is usually stated as a weighted sum of function values at specified points within the domain of integration \cite{abramowitz1988handbook,golub1969calculation,lether1978construction}. One of the famous forms of the quadrature rule is the Gauss-Legendre quadrature, which can approximate the integral of a function $f(x)$ in $[-1, 1]$ as \cite{abramowitz1988handbook,golub1969calculation},
\begin{equation}
\label{eq:47}
\int_{-1}^{1} f(x) \,dx \approx \sum_{i=1}^{N_{q}} \omega_{i} f\left(x_{i}\right),
\end{equation}
where $\omega_{i}$ are given by \cite{abramowitz1988handbook},
\begin{equation}
\label{eq:48}
\omega_{i}=\frac{2}{\left(1-x_{i}^{2}\right)\left[P_{N_{q}}^{\prime}\left(x_{i}\right)\right]^{2}}.
\end{equation}
The associated orthogonal polynomials, denoted above by $P_{N_{q}}(x)$, are referred to as Legendre polynomials, with the n-th polynomial normalized in a such way that $P_{N_{q}}(1)=1$. In particular, the $i$-th Gauss node, i.e., $x_{i}$, is the $i$-th root of $P_{N_{q}}$. Eq. (\ref{eq:47}) can be extended to a generic interval $[a,b]$ as \cite{abramowitz1988handbook},
\begin{equation}
\label{eq:49}
\begin{aligned}
\int_{a}^{b} f(x) \,dx &=\frac{b-a}{2} \int_{-1}^{1} f\left(\frac{b-a}{2} t+\frac{a+b}{2}\right) \,dt,\\
&\approx \frac{b-a}{2} \sum_{i=1}^{N_{q}} w_{i} f\left(\frac{b-a}{2} t_{i}+\frac{a+b}{2}\right).
\end{aligned}
\end{equation}
The key assumption central to the use of the Gauss-Legendre quadrature method is that the integrand $f(x)$ should be finite within the domain of integration, i.e. $\left|f(x)\right|<\infty$ for $x\in \left[a,b\right]$. The integrands in (\ref{eq:112}) meet this condition; it is easy to verify that  $\textbf{num}(\beta_{s})\neq 0$, where $\textbf{num}(\cdot)$ denotes the numerator of the fractional argument. By employing (\ref{eq:49}), the relation in (\ref{eq:112}) can be approximated as
\begin{equation}
\label{eq:84}
\begin{aligned}
R_{\mathbf{y}}(i,j) &= R_{\mathbf{y}}(l) \approx \frac{e^{\frac{-d^{2}}{p_{0}+p_{l}}}}{\pi\sqrt{\left(p_{0}^{2}-p_{l}^{2}\right)}}\left\{\int_{0}^{\frac{\pi}{2}} \frac{1}{\beta_{s}} \,d\theta \right.\\& \left. -\frac{\pi}{4}\sum_{i=1}^{N_{q}} \omega_{i}D_{1}\left(\frac{\pi}{4}(\theta_{i}+1);p_{0},p_{l},d\right) \right.\\& \left. +\frac{\pi}{4}\sum_{i=1}^{N_{q}} \omega_{i}D_{2}\left(\frac{\pi}{4}(\theta_{i}+1);p_{0},p_{l},d\right)\right\}-1,
\end{aligned}
\end{equation}
where $\theta_{i}$ denotes the $i$-th Gauss node. Note that the first part of the above integration was readily given in closed-form in~(\ref{eq:18}). 
\begin{figure}[t]
	\center{\includegraphics[width=0.6\textwidth]{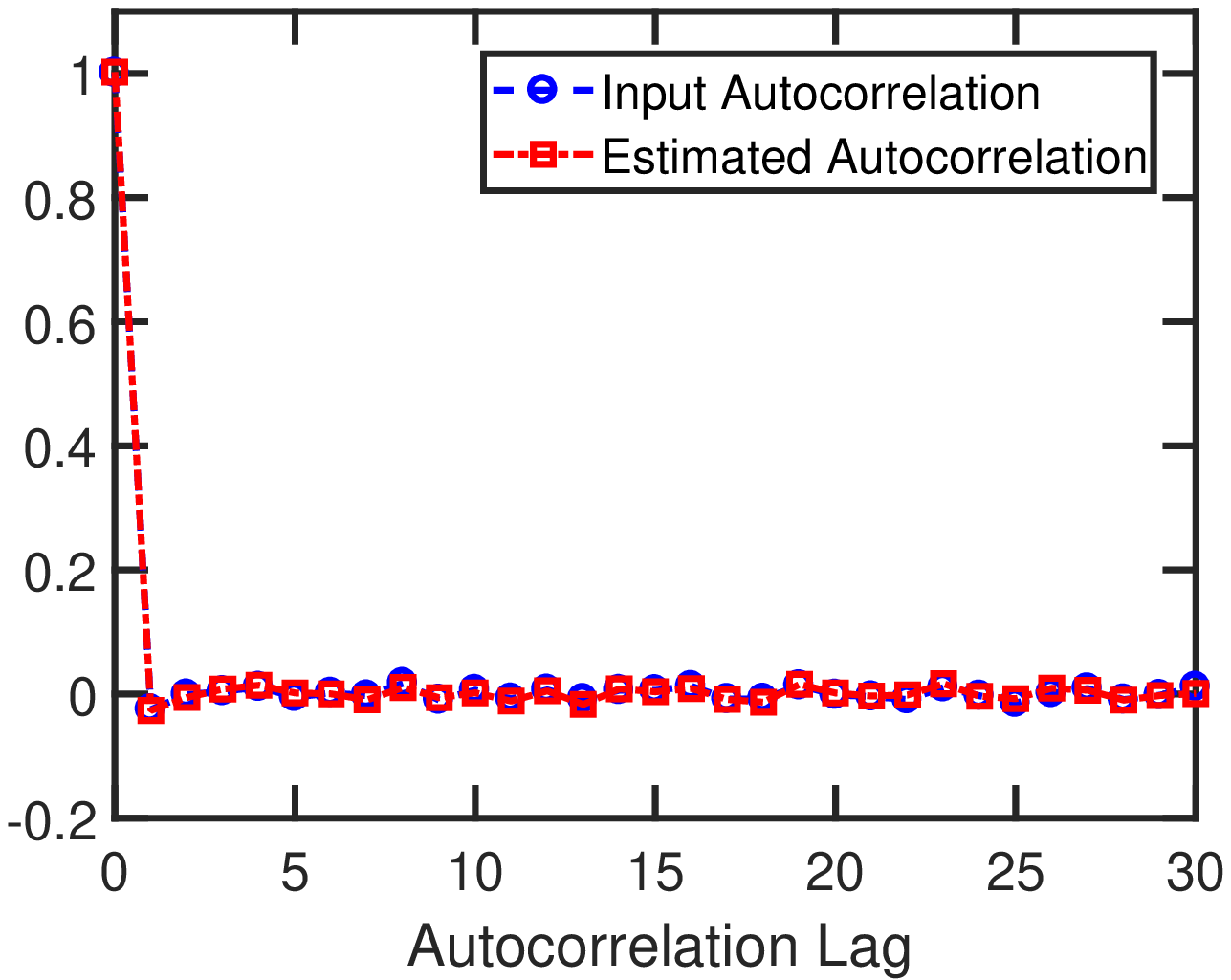}}
	\caption{Recovery of the input signal autocorrelation for a sequence of length $31$ from one-bit sampled data using the fast PA-based recovery algorithm, with the true values plotted along the estimates.}
	\label{fig_fast_3}
\end{figure}
\subsection{Recovery Criterion}
\label{subsec:6}
Based on our discussion in Section~\ref{sec:fast_recovery}, at first $p_{0}^{\star}$ is obtained by (\ref{eq:fast_1}). Then, $\{p_{l}\}$ is estimated by formulating a minimization problem; namely, we consider the following criterion ($p_{0}^{\star}=u$):
\begin{equation}
\label{eq:89}
\begin{aligned}
\bar{\Phi}(p_{l}) &\triangleq \log \left(\left|R_{\mathbf{y}}(l)-\frac{e^{\frac{-d^{2}}{u+p_{l}}}}{\pi\sqrt{\left(u^{2}-p_{l}^{2}\right)}}\left\{\int_{0}^{\frac{\pi}{2}} \frac{1}{\beta_{s}} \,d\theta \right.\right.\right.\\& \left.\left.\left. -\frac{\pi}{4}\sum_{i=1}^{N_{q}} \omega_{i}D_{1}\left(\frac{\pi}{4}(\theta_{i}+1);u,p_{l},d\right)\right.\right.\right.\\& \left.\left.\left.+\frac{\pi}{4}\sum_{i=1}^{N_{q}} \omega_{i}D_{2}\left(\frac{\pi}{4}(\theta_{i}+1);u,p_{l},d\right)\right\}+1\right|^{2}\right).
\end{aligned}
\end{equation}
By now, we have derived an approximated version of (\ref{eq:112}) using the Gauss-Legendre quadrature. Let $J_{s}(p_{l})$ denote this approximation. Therefore, we can alternatively consider the criterion:
\begin{equation}
\label{eq:91}
\Phi(p_{l}) \triangleq \log\left(\left|R_{\mathbf{y}}(l)-J_{s}(p_{l})\right|^{2}\right).
\end{equation}
Surprisingly, the criterion in (\ref{eq:91}) appears to be a convex function with respect to $p_{l}$ (a proof is presented in Appendix~B)---see Fig.~\ref{figure_12} for an example of the optimization landscape of $\Phi(p_{l})$. By considering the feasible region of $\{p_{l}\}$, the following problem is cast:
\begin{equation}
\label{eq:93}
\begin{aligned}
\mathcal{P}_{\ell}&: &\min_{p_{l}}& &\Phi(p_{l}),& &\mbox{s.t.}& &-u \leq p_{l} \leq u.
\end{aligned}
\end{equation}
The convex problem in (\ref{eq:93}) may be solved efficiently via the golden section search and parabolic interpolation approach.
Once $\{p_{l}\}$ is obtained, one can estimate the input autocorrelation values $\{r_{l}\}$ via \eqref{eq:10}. The acquired optimum recovery results will be presented in the following.

\begin{figure}[t]
	\center{\includegraphics[width=0.6\textwidth]{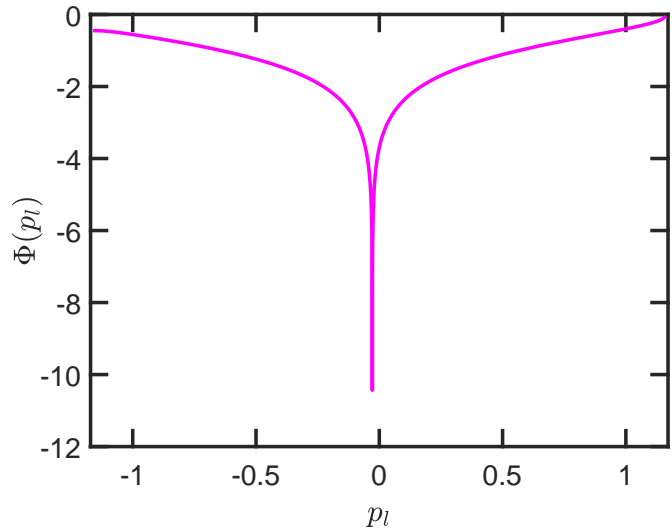}}
	\caption{Example plot of the Gauss-Legendre quadrature approach-based estimation criterion $\Phi(p_{l})$ with respect to $p_{l}$ showing its convexity.}
	\label{figure_12}
\end{figure}




\subsection{Numerical Results}
\label{subsec:18}
We will now examine the Gauss-Legendre quadrature approach by comparing its recovery results with the true input signal autocorrelation values. In all experiments, the input signals were generated as zero-mean Gaussian sequences with unit variance. Accordingly, we made use of the time-varying thresholds with $d=0.3$ and a diagonal matrix $\bSigma$ whose diagonal entries are set to $0.1$.

\begin{figure}[t]
	\center{\includegraphics[width=0.6\textwidth]{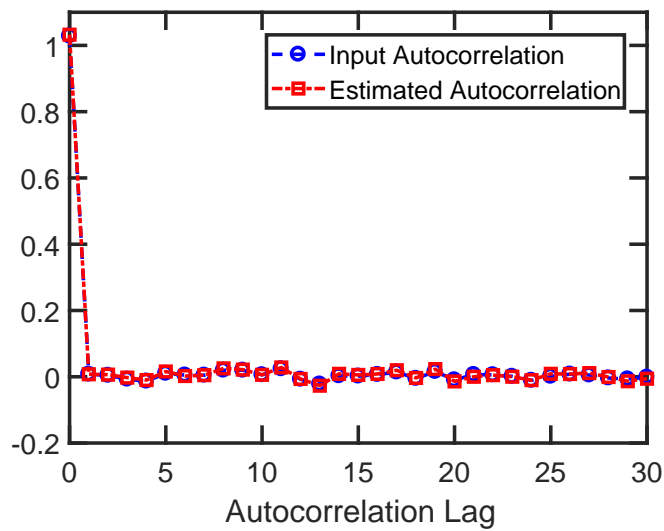}}
	\caption{Recovery of the input signal autocorrelation for a sequence of length $31$ from one-bit sampled data using the Gauss-Legendre quadrature approach, with the true values plotted along the estimates.}
	\label{figure_8}
\end{figure}



To show the effectiveness of the Gauss-Legendre quadrature approach, we present an example of autocorrelation sequence recovery. The true input signal autocorrelation and the estimated autocorrelation values by this approach are shown in Fig.~\ref{figure_8} for a random sequence of length $31$. Fig.~\ref{figure_8} appears to confirm the possibility of recovering the autocorrelation values in our example. The number of quadrature points $N_{q}$ is considered to be $13$ based on our model performance.


\section{Monte-Carlo Integration for Covariance Recovery}
In this section, another well-known approach referred to as the Monte-Carlo integration is utilized to evaluate the integral in~(\ref{eq:112}); as deemed essential for the recovery of $\{p_{l}\}$ since $p_{0}$ is obtained by (\ref{eq:fast_1}). We begin by a brief overview of the Monte-Carlo integration method  in Section \ref{subsec:9}. We then move to formulate an approximated version of (\ref{eq:112}) based on the Monte-Carlo integration technique. Lastly, the efficacy of this approach in estimating the input autocorrelation values is numerically evaluated.

\subsection{The Monte-Carlo Integration Method: An Overview}
\label{subsec:9}
The Monte-Carlo integration is another extensively used approach in numerical analysis to approximate the definite integral of a function, stated as an expectation of the function over uniform random variables as below \cite{evans2000approximating,caflisch1998monte}:
\begin{equation}
\label{eq:50}
\begin{aligned}
\mathbb{E}\left\{f(x)\right\} ~\,&=\int^{b}_{a} f(x)p(x)\,dx\approx \frac{1}{N_{m}}\sum^{N_{m}}_{i=1}f(x_{i}),\\ p(x)=\frac{1}{b-a} &\Rightarrow \int^{b}_{a} f(x) \,dx \approx \frac{b-a}{N_{m}} \sum_{i=1}^{N_{m}} f(x_{i}),
\end{aligned}
\end{equation}
where $p(x)=\frac{1}{b-a}$ is the uniform probability distribution in the interval $\left[a,b\right]$. 
By employing (\ref{eq:50}), the expression in (\ref{eq:112}) may be rewritten as
\begin{equation}
\label{eq:98}
\begin{aligned}
R_{\mathbf{y}}(i,j) = R_{\mathbf{y}}(l) &\approx \frac{e^{\frac{-d^{2}}{p_{0}+p_{l}}}}{\pi\sqrt{\left(p_{0}^{2}-p_{l}^{2}\right)}}\left\{\int_{0}^{\frac{\pi}{2}} \frac{1}{\beta_{s}} \,d\theta \right.\\&\quad  \left. -\frac{\pi}{2N_{m}}\sum_{i=1}^{N_{m}} D_{1}\left(\theta_{i};p_{0},p_{l},d\right) \right.\\& \quad \left. +\frac{\pi}{2N_{m}}\sum_{i=1}^{N_{m}} D_{2}\left(\theta_{i};p_{0},p_{l},d\right)\right\}-1,
\end{aligned}
\end{equation}
where $\theta_{i}$ denotes the $i$-th random number generated from the uniform distribution in the interval $\left[0,\frac{\pi}{2}\right]$. Note that the first part of the above integral was readily evaluated in closed-form in~(\ref{eq:18}).

\subsection{Recovery Criterion}
\label{subsec:10}
Similar to two previous cases, at first $p_{0}^{\star}$ is obtained by (\ref{eq:fast_1}). Then, the parameter of interest, i.e., $\{p_{l}\}$, is estimated by formulating a minimization problem. Namely, we consider the following criterion ($u=p_{0}^{\star}$):
\begin{equation}
\label{eq:99}
\begin{aligned}
\bar{\Omega}(p_{l}) &\triangleq \log \left(\left|R_{\mathbf{y}}(l)-\frac{e^{\frac{-d^{2}}{u+p_{l}}}}{\pi\sqrt{\left(u^{2}-p_{l}^{2}\right)}}\left\{\int_{0}^{\frac{\pi}{2}} \frac{1}{\beta_{s}} \,d\theta \right.\right.\right.\\& \left.\left.\left. -\frac{\pi}{2N_{m}}\sum_{i=1}^{N_{m}} D_{1}\left(\theta_{i};u,p_{l},d\right)\right.\right.\right.\\& \left.\left.\left.+\frac{\pi}{2N_{m}}\sum_{i=1}^{N_{m}} D_{2}\left(\theta_{i};u,p_{l},d\right)\right\}+1\right|^{2}\right),
\end{aligned}
\end{equation}
where the autocorrelation of output signal $R_{\mathbf{y}}$ is estimated via (\ref{eq:23}). Suppose an approximated version of (\ref{eq:112}) has been obtained using the Monte-Carlo integration approach, which is denoted by $F_{s}(p_{l})$. Thus, the above criterion may be approximated via the following:
\begin{equation}
\label{eq:100}
\Omega(p_{l}) \triangleq \log\left(\left|R_{\mathbf{y}}(l)-F_{s}(p_{l})\right|^{2}\right).
\end{equation}
Similar to the previous criterion (\ref{eq:91}), $\Omega(p_{l})$ appears to be a convex function respect to $p_{l}$, whose proof of convexity is similar to that of $\Phi(.)$ in Appendix~B---see Fig.~\ref{figure_14} for an example of the optimization landscape of $\Omega(p_{l})$. By considering the feasible region of $\{p_{l}\}$, we can formulate the following recovery problem:
\begin{equation}
\label{eq:101}
\begin{aligned}
\mathcal{P}_{\ell}&: &\min_{p_{l}}& &\Omega(p_{l}),& &\mbox{s.t.}& &-u \leq p_{l} \leq u,
\end{aligned}
\end{equation}
which may be tackled using the same tools as proposed in Section~\ref{subsec:6}. The recovery of $\{p_{l}\}$ results in estimating the autocorrelation values of $\mathbf{x}$ via (\ref{eq:10}). The obtained recovery results will be investigated in the following.

\begin{figure}[t]
	\center{\includegraphics[width=0.6\textwidth]{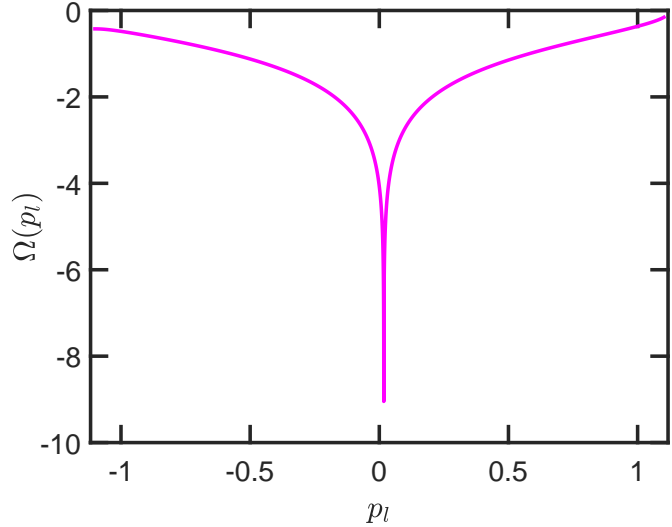}}
	\caption{Example plot of the estimation criterion $\Omega(p_{l})$ with respect to $p_{l}$ showing its convexity.}
	\label{figure_14}
\end{figure}

\subsection{Numerical Results}
\label{subsec:17}
We will examine the Monte-Carlo integration approach by comparing its recovery results with the true input signal autocorrelation values. In all experiments, the input signals were generated as zero-mean Gaussian sequences with unit variance. Accordingly, we made use of the time-varying thresholds with $d=0.3$ and a diagonal matrix $\bSigma$ whose diagonal entries are set to $0.1$.

To show the efficacy of the Monte-Carlo-based approach, we compare the input signal autocorrelation values for $31$ lags with the true values as presented in Fig.~\ref{figure_15}. The number of nodes ($N_{m}$) was experimentally set to $2000$ based on our model error.

\section{Comparing the Proposed Recovery Methods}
We will now compare all the discussed approaches in the autocorrelation sequence recovery for stationary signals. We will take advantage of (\ref{eq:fast_1}) to obtain the optimal value of $p_{0}$ in (\ref{eq:112}). To recover the desired parameter $\left\{p_{l}\right\}$ for $l\neq0$, we presented three approaches: (i) employing the \text{Padé} approximation of the integrands in (\ref{eq:112}), also referred to as the PA technique, (ii) applying the Gauss-Legendre quadrature technique, and (iii) applying the Monte-Carlo integration to evaluate the integral in (\ref{eq:112}). As was observed before, all three approaches show promising recovery results. To numerically compare these approaches, we consider input signals $\mathbf{x}\in \mathbb{R}^{5}$ generated as zero-mean Gaussian sequences with unit variance. The time-varying threshold setting is the same as Section~\ref{subsec:18}. As a metric for comparisons, we use the experimental mean square error (MSE) of an estimate $\hat r_{l}$ of an autocorrelation value $r_{l}$, defined as
\begin{equation}
\label{eq:4600000}
\begin{aligned}
\mathrm{MSE} \triangleq \frac{1}{EL} \sum^{E}_{e=1}\sum^{L}_{l=1} |r^{e}_{l}-{\hat r^{e}_{l}}|^{2},
\end{aligned}
\end{equation}
where $\left\{r^{e}_{l},\hat r^{e}_{l}\right\}$ are the autocorrelation values and their estimates at the $e$-th experiment, with the number of lags set to $L=4$. The number of experiments is assumed to be $E=5$. The results are obtained for the number of ensembles $N_{\mathbf{x}}\in\left\{1000, 3000, 6000, 10000\right\}$.

Fig.~\ref{figure_66} shows that the Gauss-Legendre method has a better performance in recovering the input signal autocorrelation values in comparison with the PA technique and the Monte-Carlo integration. Other than the PA-based recovery, the proposed numerical approaches are capable of recovering the input autocorrelation values via convex programming, which makes them appealing. Nevertheless, the proper selection of the number of nodes and quadrature points in the Gauss-Legendre quadrature and the Monte-Carlo integration techniques is crucial and may present itself as a bottleneck in an effective recovery.

\vspace{5pt}

\underline{\emph{Remark:}} Since the true input signal autocorrelation values are unknown \emph{a priori}, the above observation hints at the practical value of the PA technique. On the other hand, one can use the outcome of the PA technique as an initial estimate, to subsequently run the other slightly improved approximation techniques. 

\begin{figure}[t]
	\center{\includegraphics[width=0.6\textwidth]{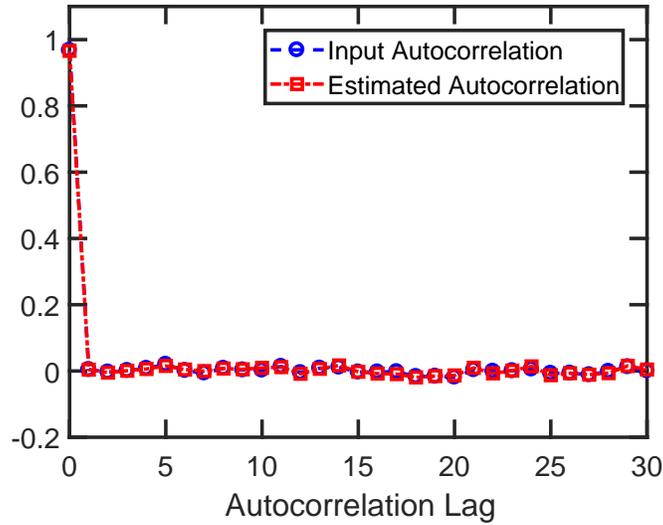}}
	\caption{Recovery of the input signal autocorrelation for a sequence of length $31$ from one-bit sampled data using the Monte-Carlo integration approach, with the true values plotted alongside the estimates.}
	\label{figure_15}
\end{figure}

\section{Modified Bussgang Law for Time-Varying Sampling Thresholds}
In addition to the arcsine law, the Bussgang law unveils an important connection in stochastic analysis of one-bit correlation data. It states that the cross-correlation of a Gaussian signal before and after it has passed through the nonlinear sampling operation is equal to its autocorrelation up to a constant \cite{Bussgang1952crosscorrelation}. In this section, at first, we review the original Bussgang law and its formalism for one-bit quantization systems. Secondly, a modified Bussgang law is presented for cases where the input signals are sampled using time-varying thresholds.

\begin{figure}[t]
	\center{\includegraphics[width=0.6\textwidth]{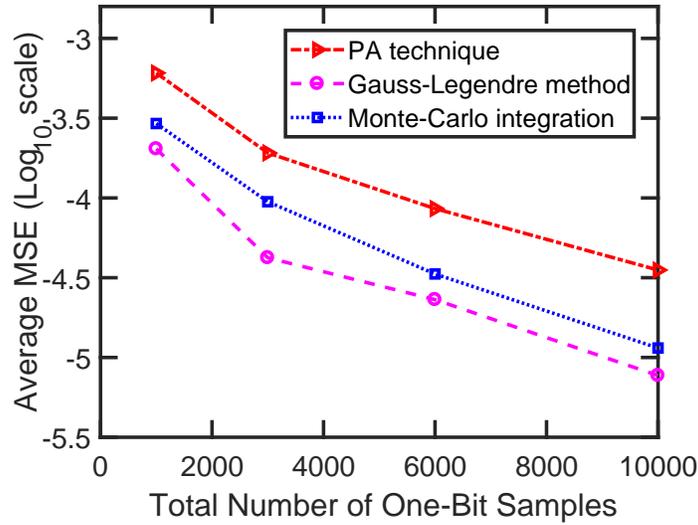}}
	\caption{Comparing the three proposed methods (PA technique when we use the fast covariance recovery formulation (\ref{eq:fast_2}), Gauss-Legendre method and Monte-Carlo integration) in recovering the input stationary signal autocorrelation by average obtained  MSE for different one-bit sample sizes.}
	\label{figure_66}
\end{figure}

\subsection{The Bussgang Law for One-Bit Quantization}
The \emph{Bussgang law} informs on the second order statistics of one-bit sampled data by relating the cross-correlation function of a stationary zero-mean Gaussian input signal $\mathbf{x}$ and the output $\mathbf{y}$ of a nonlinear memoryless amplitude-distortion function with the autocorrelation function of the input signal as follows~\cite{Bussgang1952crosscorrelation}:
\begin{equation}
\label{eq:73}
\bR_{\mathbf{xy}}=C \bR_{\mathbf{x}},
\end{equation}
where $R_{\mathbf{xy}}$ is the cross-correlation function between input and output signals of the nonlinear system ($\mathbf{y}=g(\mathbf{x})$ where $g(.)$ is the nonlinear memoryless amplitude-distortion function). Also, $C$ is defined as \cite{Bussgang1952crosscorrelation},
\begin{equation}
\label{eq:74}
C=\frac{1}{ \sqrt{2 \pi R^{3}_{\mathbf{x}}(0)}} \int_{-\infty}^{\infty} x_{i} g(x_{i}) e^{-\frac{x_{i}^{2}}{2 R_{\mathbf{x}}(0)}} \,dx_{i},
\end{equation}
where $x_{i}$ is the $i$-th entry of $\mathbf{x}$. If we consider $g(.)$ to be a sign function, 
we have a one-bit quantization system and $C$ is obtained as,
\begin{equation}
\label{eq:75}
\begin{aligned}
C 
= \sqrt{\frac{2}{\pi}}\; R_{\mathbf{x}}^{-\frac{1}{2}}(0).
\end{aligned}
\end{equation}

\subsection{The Modified Bussgang Law}
When we consider time-varying thresholds, the cross-correlation matrix between the one-bit sampled signal and the input signal can be written in the following form.
\begin{theorem}
\label{theorem_3}
Suppose $\btau\sim\mathcal{N}\left(\mathbf{d}=\mathbf{1}d,\bSigma\right)$ is a time-varying threshold, and $\mathbf{x}$ is a stationary input signal. Let $\mathbf{y}=g\mathbf{(w)}$ denote the one-bit sampled data, where $\mathbf{w}=\mathbf{x}-\btau$, with $p_{0}$ denoting its associated variance. Then, the cross correlation matrix between $\mathbf{y}$ and $\mathbf{x}$ satisfies the relation,
\begin{equation}
\label{eq:140}
\bR_{\mathbf{yx}}=\bR_{\mathbf{y}\btau}+ \left[C_{1}\left(\bR_{\mathbf{x}}+\Sigma\right)+ C_{2}d\left(\bR_{\mathbf{x}}+\Sigma-p_{0}\mathcal{U}\right)\right],
\end{equation}
where $\mathcal{U}$ is an all-one matrix, and $C_{1}$ and $C_{2}$ are given by
\begin{equation}
\label{eq:141}
\begin{aligned}
C_{1} &= \sqrt{\frac{2}{\pi p_{0}}}\Gamma\left(1,\dfrac{d^2}{2p_0}\right)-\frac{d}{\sqrt{\pi p^{2}_{0}}}\left(\Gamma\left(\dfrac{1}{2},\dfrac{d^2}{2p_0}\right)-\sqrt{{\pi}}\right),\\ C_{2} &= -\frac{1}{p_{0}} \operatorname{erf}\left(\frac{d}{\sqrt{2p_{0}}}\right).
\end{aligned}
\end{equation}
\end{theorem}
\begin{proof}
Suppose $w_i$ and $w_j$ are the $i$-th and $j$-th entries of $\mathbf{w}$ ($i\neq j$) with $\mathbb{E}\{w_{i}\}=\mathbb{E}\{w_{i}\}=-d$, and that $p_{l}$ and $p_{0}$ denote the autocorrelation term for lag $l=|i-j|$ and variance of~$\mathbf{w}$, respectively. Consider the quantized random variables $y_{i}=g(w_{i})$ and $y_{j}=g(w_{i})$, where $g(.)$ denotes the non-linear transformation function. The cross-correlation function between $w_{i}$ and $y_{j}$ can thus be obtained as below:
\begin{equation}
\label{eq:76}
\begin{aligned}
R_{\mathbf{yw}}(i,j) &= \frac{1}{2\pi\sqrt{p_{0}^2-p_{l}^2}} \int^{\infty}_{-\infty}\int^{\infty}_{-\infty} w_{i}g(w_{j})e^{\lambda(d)}\,dw_{i}\,dw_{j}
\end{aligned}
\end{equation}
where $\lambda(d)$ is defined in  (\ref{eq:5}). We begin by evaluating the integral in (\ref{eq:76}) with respect to $w_{i}$ as,
\begin{equation}
\label{eq:80}
\begin{aligned}
R_{\mathbf{yw}}(i,j) &= \frac{e^{\frac{-d^{2}}{p_0+p_{l}}}}{2\pi\sqrt{p_{0}^2-p_{l}^2}} \int^{\infty}_{-\infty}  g(w_{j})e^{\frac{2d(p_{0}-p_{l})w_{j}+w^{2}_{j}p_{0}}{-2(p^{2}_{0}-p^{2}_{l})}}\\ &\int^{\infty}_{-\infty}w_{i}e^{\frac{2d(p_{0}-p_{l})w_{i}+w^{2}_{i}p_{0}-2p_{ij}w_{i}w_{j}}{-2(p^{2}_{0}-p^{2}_{l})}} \,dw_{i}\,dw_{j}\\ &=  C_{1}p_{l}-C_{2} d(p_{0}-p_{l}),
\end{aligned}
\end{equation}
where $C_{1}$ and $C_{2}$ are given by
\begin{equation}
\label{eq:81}
\begin{aligned}
C_{1} &= \frac{1}{\sqrt{2\pi p^{3}_{0}}}\int^{\infty}_{-\infty} w_{j}g(w_{j})e^{-\frac{(w_{j}+d)^{2}}{2p_{0}}}\,dw_{j},\\
C_{2} &= \frac{1}{\sqrt{2\pi p^{3}_{0}}}\int^{\infty}_{-\infty}g(w_{j})e^{-\frac{(w_{j}+d)^{2}}{2p_{0}}}\,dw_{j}.
\end{aligned}
\end{equation}
A detailed proof of the results in  (\ref{eq:80}) and (\ref{eq:81}) is presented in Appendix~C. Next note that  (\ref{eq:80}) can be rewritten as
\begin{equation}
\label{eq:82}
\begin{aligned}
\bR_{\mathbf{yw}} &= C_{1}\bR_{\mathbf{w}}-d C_{2}\left(p_{0}\mathcal{U}-\bR_{\mathbf{w}}\right),
\end{aligned}
\end{equation}
where $\mathcal{U}$ is an all-one matrix,  and $\bR_{\mathbf{yw}}$ can be simplified as
\begin{equation}
\label{eq:83}
\begin{aligned}
\bR_{\mathbf{yw}} &= \mathbb{E}\{\mathbf{y}(\mathbf{x}-\btau)^{\mathrm{H}}\},\\
&= \mathbb{E}\{\mathbf{y}\mathbf{x}^{\mathrm{H}}\}-\mathbb{E}\{\mathbf{y}\btau^{\mathrm{H}}\},\\ &= \bR_{\mathbf{yx}}-\bR_{\mathbf{y}\btau}.
\end{aligned}
\end{equation}
Since the covariance matrix of $\mathbf{w}$ is $\bR_{\mathbf{w}}=\bR_{\mathbf{x}}+\Sigma$, our \emph{modified Bussgang law} will thus take the form,
\begin{equation}
\label{eq:77}
\bR_{\mathbf{yx}}-\bR_{\mathbf{y}\btau}= (C_{1}+d C_{2})\left(\bR_{\mathbf{x}}+\Sigma\right)-d C_{2}p_{0}\mathcal{U}. \end{equation}
If the nonlinear function $g(.)$ is the sign function, $C_{1}$ and $C_{2}$ are given in closed-form as,
\begin{equation}
\label{eq:130}
\begin{aligned}
C_{1} &=\frac{1}{\sqrt{2\pi p^{3}_{0}}}\int^{\infty}_{0} w_{j}\left\{e^{-\frac{(w_{j}+d)^{2}}{2p_{0}}}+e^{-\frac{(w_{j}-d)^{2}}{2p_{0}}}\right\}\,dw_{j}\\
&= \sqrt{\frac{2}{\pi p_{0}}}\Gamma\left(1,\dfrac{d^2}{2p_0}\right)-\frac{d}{\sqrt{\pi p^{2}_{0}}}\left(\Gamma\left(\dfrac{1}{2},\dfrac{d^2}{2p_0}\right)-\sqrt{{\pi}}\right),\\
C_{2} &= \frac{1}{\sqrt{2\pi p^{3}_{0}}}\int^{\infty}_{0} \left\{e^{-\frac{(w_{j}+d)^{2}}{2p_{0}}}-e^{-\frac{(w_{j}-d)^{2}}{2p_{0}}}\right\}\,dw_{j}\\
&= -\frac{1}{p_{0}} \operatorname{erf}\left(\frac{d}{\sqrt{2p_{0}}}\right),
\end{aligned}
\end{equation}
where $\Gamma(.,.)$ denotes an \emph{incomplete gamma function}~\cite{jameson2016incomplete,amore2005asymptotic}. Based on (\ref{eq:77}), the cross-correlation matrix between the input and the output one-bit data are computed where $p_{0}$ is obtained by (\ref{eq:fast_1}) and $\{p_{l}\}$ can be either recovered using (\ref{eq:fast_2}), (\ref{eq:93}) or (\ref{eq:101}). In addition, $\bR_{\mathbf{x}}$ is obtained through (\ref{eq:10}). Note that the cross-correlation matrix between the threshold vector $\btau$ and the output vector $\mathbf{y}$ can be estimated via a \emph{sample cross-correlation matrix} as,
\begin{equation}
\label{eq:85}
\bR_{\mathbf{y}\btau}\approx \frac{1}{N_{\mathbf{x}}} \sum_{k=1}^{N_{\mathbf{x}}} \mathbf{y}(k) \btau(k)^{\mathrm{H}}.
\end{equation}
\end{proof}
Note that the reliance of the cross-correlation recovery on the recovery of the autocorrelation values paves the way for the three proposed autocorrelation recovery approaches to be used as an intermediate stage for cross-correlation recovery via our modified Bussgang law. This will lead to cross-correlation recovery methods with various levels of accuracy.
\begin{figure*}[t]
	\centering
	\begin{subfigure}[b]{0.45\textwidth}
		\includegraphics[width=1\linewidth]{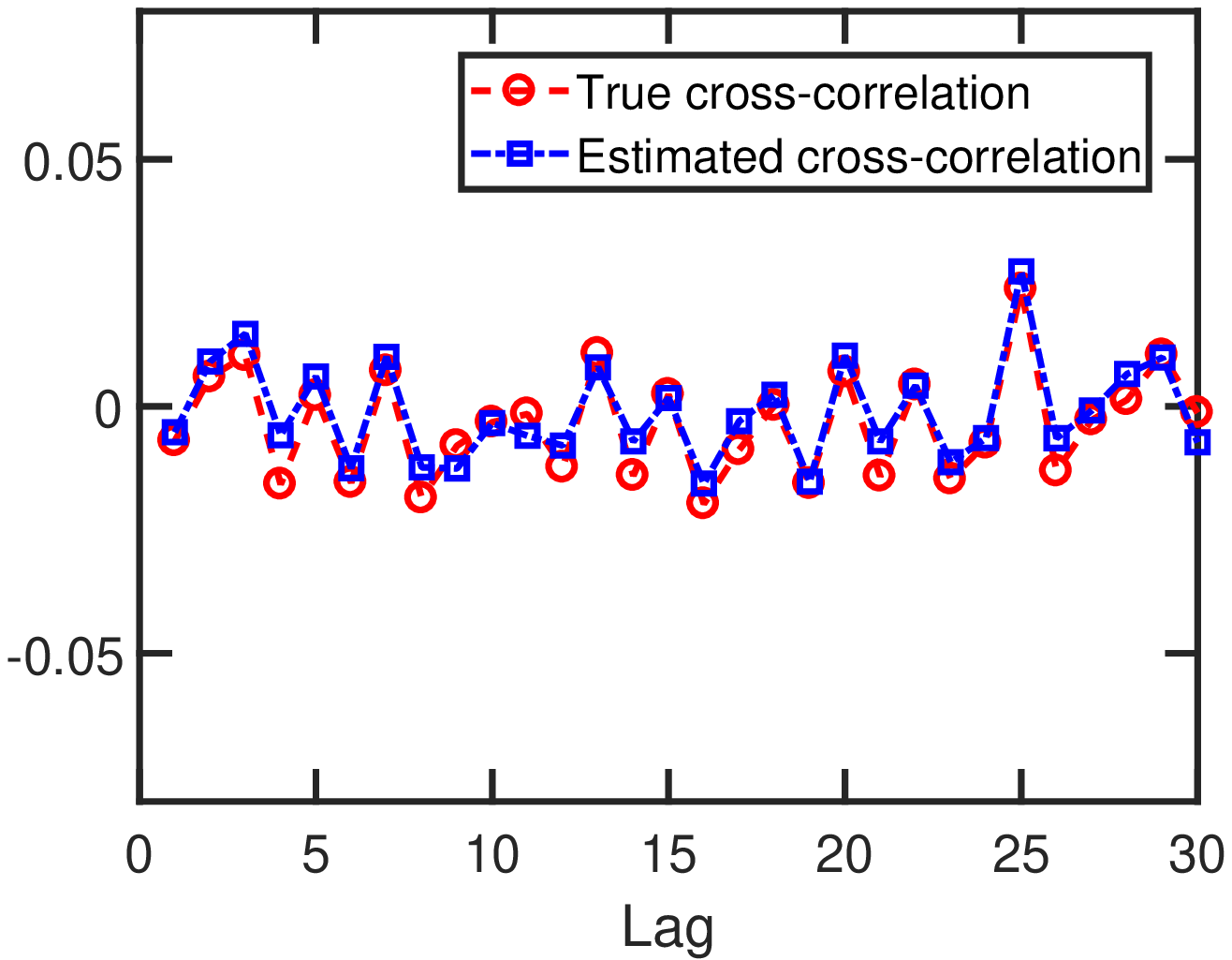}
		\caption{\text{Padé} approximation technique}
	\end{subfigure}
	\begin{subfigure}[b]{0.45\textwidth}
		\includegraphics[width=1\linewidth]{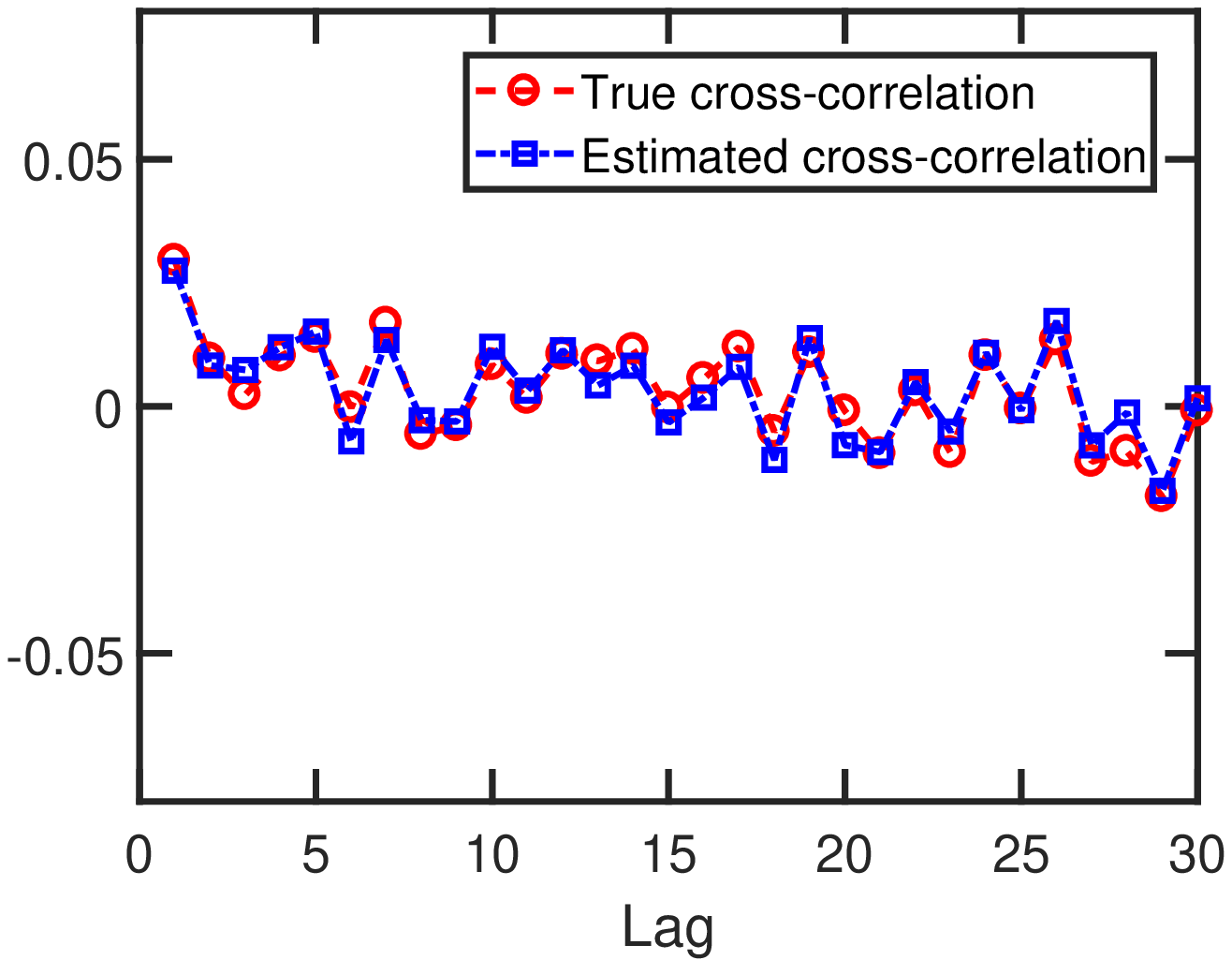}
		\caption{Gauss-Legendre quadrature technique}
	\end{subfigure}
	\begin{subfigure}[b]{0.45\textwidth}
		\includegraphics[width=1\linewidth]{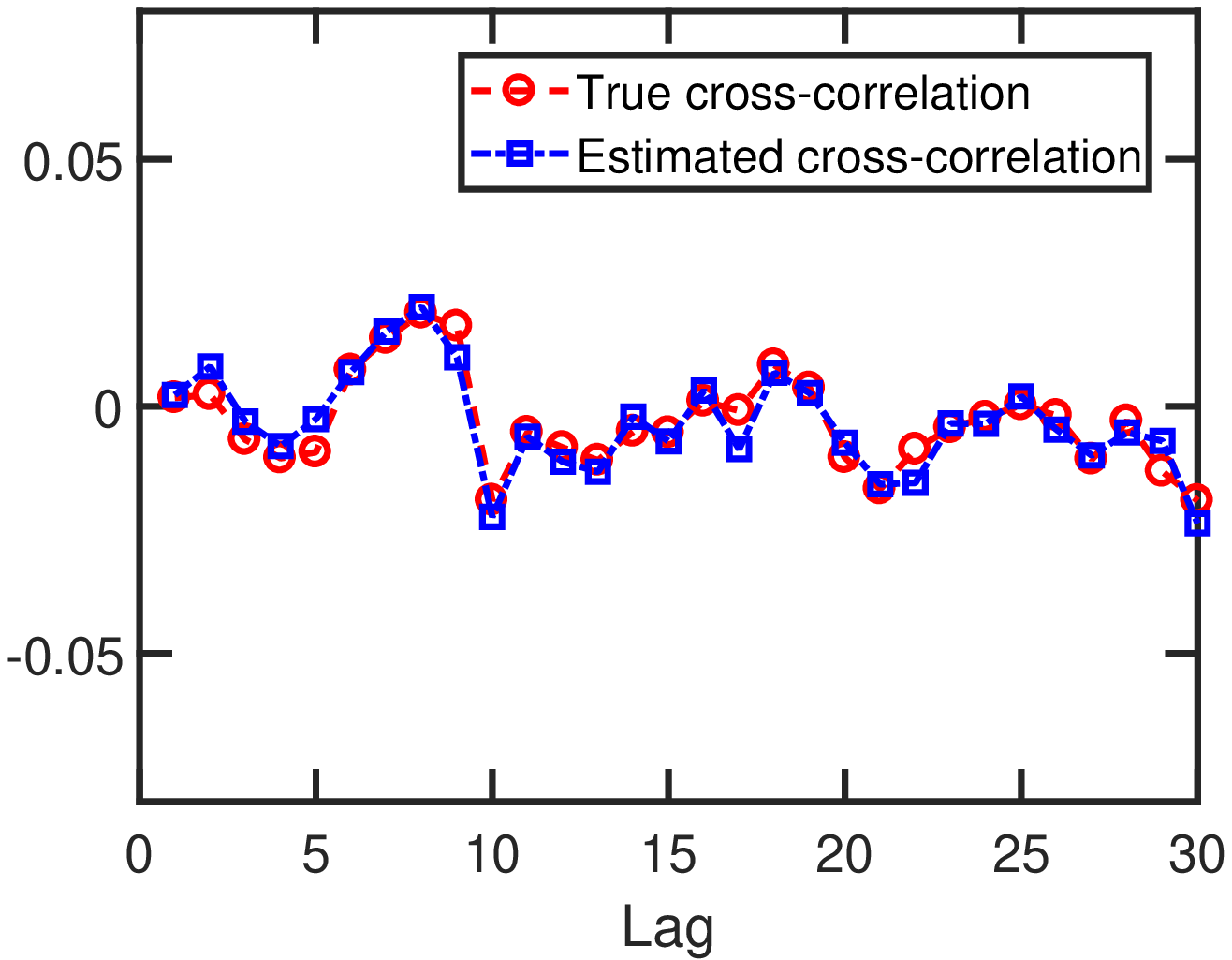}
		\caption{Monte-Carlo integration technique}
	\end{subfigure}
	\caption{The recovery of the cross-correlation between the input signal and the one-bit sampled data by the modified Bussgang law applied in conjunction with  various one-bit autocorrelation recovery approaches for a sequence of length $30$, with the true values plotted alongside the estimates.}
	\label{figure_80}
\end{figure*}
\subsection{A Numerical Investigation of the Modified Bussgang Law}
In this section, we will examine the proposed modified Bussgang law by comparing its recovery results with the true cross-correlation values between the input signal and the one-bit sampled data. In all experiments, the input signal settings are the same as in Section~\ref{sec:stationary_analytic}. The time-varying threshold settings are as follows: (a) PA: $d=0.1$ and $\bSigma=0.2\bI$, (b) Gauss-Legendre: $d=0.3$ and $\bSigma=0.1\bI$, (c) Monte-Carlo: $d=0.3$ and $\bSigma=0.1\bI$, where $\bI$ denotes the identity matrix.

In order to showcase the effectiveness of the proposed approach, we present an example of cross-correlation sequence recovery. The true cross-correlation between the input signal and the one-bit sampled data and the estimated cross-correlation values by our approach are shown in Fig.~\ref{figure_80} for a random sequence of length $30$. Fig.~\ref{figure_80} appears to confirm the possibility of recovering the cross-correlation values from one-bit sampled data with time-varying thresholds by employing any of the three recovery methods (PA, Gauss-Legendre method and Monte-Carlo integration). The difference between the true values and the estimated values in Fig.~\ref{figure_80} is presumably for the most part due to the numerical approximations used for the error function, and the incomplete gamma function utilized in (\ref{eq:140}). In addition, estimation error in the autocorrelation recovery used to estimate the desired parameters $p_{0}$ from (\ref{eq:fast_1}) and $p_{l}$ from (\ref{eq:fast_2}), (\ref{eq:93}) and (\ref{eq:101}), can propagate to the cross-correlation recovery as well.
\section{Conclusion}
We proposed a modified arcsine law  that can make use of non-zero time-varying thresholds in one-bit sampling when the input signal is assumed to be stationary. Our extended results take advantage of \text{Padé} approximations, as well as numerical approaches such as the Gauss-Legendre  and the Monte-Carlo integration techniques. The numerical results showcased the effectiveness of the proposed approaches in recovering the autocorrelation values. 
We finalized our work  by proposing a modified Bussgang law for one-bit sampling of stationary  signals with time-varying thresholds.
\appendices
\section{Detailed Derivations for the Integral in (\ref{eq:7})}
The focus herein is on obtaining the ultimate formalism for $R_{\mathbf{y}}(i,j)$ in (\ref{eq:9}) from the relation in (\ref{eq:7}). 
In particular, based on (\ref{eq:7}), we define $\zeta(\alpha_{s},\beta_{s})
\triangleq \int_{0}^{\infty} e^{-\beta_{s}\rho^2}\left(e^{-\alpha_{s}\rho}+e^{\alpha_{s}\rho}\right)\rho\,d\rho$ and simplify it as,
\begin{equation}
\label{eq:59}
\begin{aligned}
\zeta(\alpha_{s},\beta_{s}) &= \int_{0}^{\infty}\left(e^{-\beta_{s} \rho^{2}+\alpha_{s} \rho}+e^{-\beta_{s} \rho^{2}-\alpha_{s} \rho}\right)\rho \,d\rho\\ &= \int_{0}^{\infty} e^{\frac{\alpha^{2}_{s}}{4 \beta_{s}}}\left(e^{-\beta_{s}\left(\rho^{2}+\frac{\alpha_{s}}{\beta_{s}} \rho+\frac{a^{2}}{4 \beta^{2}_{s}}\right)}\right.\\ &\left.+e^{-\beta_{s}\left(\rho^{2}-\frac{\alpha_{s}}{\beta_{s}} \rho+\frac{\alpha^{2}_{s}}{4 \beta^{2}_{s}}\right)}\right)\rho\,d\rho\\ &=\int_{0}^{\infty} e^{\frac{\alpha^{2}_{s}}{4 \beta_{s}}}\left(e^{-\beta_{s}\left(\rho+\frac{\alpha_{s}}{2\beta_{s}}\right)^{2}}+e^{-\beta_{s}\left(\rho-\frac{\alpha_{s}}{2\beta_{s}}\right)^{2}}\right)\rho \,d\rho.
\end{aligned}
\end{equation}
We can now split the integration in (\ref{eq:59}) into two parts as below:
\begin{equation}
\label{eq:60}
\begin{aligned}
\zeta(\alpha_{s},\beta_{s}) &= e^{\frac{\alpha^{2}_{s}}{4 \beta_{s}}} \int_{0}^{\infty} e^{-\beta_{s}\left(\rho+\frac{\alpha_{s}}{2\beta_{s}}\right)^{2}} \rho \,d\rho\\ &+e^{\frac{\alpha^{2}_{s}}{4 \beta_{s}}} \int_{0}^{\infty} e^{-\beta_{s}\left(\rho-\frac{\alpha_{s}}{2\beta_{s}}\right)^{2}} \rho \,d\rho\\ &= e^{\frac{a^{2}}{4 \beta_{s}}} \int_{\frac{\alpha_{s}}{\beta_{s} 2}}^{\infty} e^{-\beta_{s}(a)^{2}}\left(a-\frac{\alpha_{s}}{2\beta_{s}}\right)\,da\\ &+e^{\frac{\alpha^{2}_{s}}{4 \beta_{s}}} \int_{-\frac{a}{2\beta_{s}}}^{\infty} e^{-\beta_{s}(a)^{2}}\left(a+\frac{\alpha_{s}}{\beta_{2} 2}\right) \,da\\ &= \mathbb{I}_{1}+\mathbb{I}_{2},
\end{aligned}
\end{equation}
where $\mathbb{I}_{1}$ is constructed as,
\begin{equation}
\label{eq:61}
\begin{aligned}
\mathbb{I}_{1} &= e^{\frac{\alpha^{2}_{s}}{4 \beta_{s}}} \int_{\frac{\alpha_{s}}{2\beta_{s}}}^{\infty} e^{-\beta_{s}(a)^{2}}a \,da-\frac{\alpha_{s}}{2\beta_{s}} e^{\frac{\alpha^{2}_{s}}{4 \beta_{s}}} \int_{\frac{\alpha_{s}}{2\beta_{s}}}^{\infty} e^{-\beta_{s}(a)^{2}} \,da\\
&= \frac{1}{2\beta_{s}} e^{\frac{\alpha^{2}_{s}}{4 \beta_{s}}} \int_{\frac{\alpha^{2}_{s}}{4\beta_{s}}}^{\infty} e^{-u} \,du-\sqrt{\frac{\pi}{\beta_{s}}} \frac{\alpha_{s}}{2\beta_{s}} e^{\frac{\alpha^{2}_{s}}{4 \beta_{s}}} \frac{1}{\sqrt{2 \pi}} \int_{\frac{\alpha_{s}}{\sqrt{2 \beta_{s}}}}^{\infty} e^{-\frac{u^{2}}{2}} \,du\\ &= \frac{1}{2 \beta_{s}}-\sqrt{\frac{\pi}{\beta_{s}}} \frac{\alpha_{s}}{\beta_{s} 2} e^{\frac{\alpha^{2}_{s}}{4 \beta_{s}}} Q\left(\frac{\alpha_{s}}{\sqrt{2 \beta_{s}}}\right).
\end{aligned}
\end{equation}
Similar to above process, we have $\mathbb{I}_{2}=\frac{1}{2 \beta_{s}}+\sqrt{\frac{\pi}{\beta_{s}}} \frac{\alpha_{s}}{2\beta_{s}} e^{\frac{\alpha^{2}_{s}}{4 \beta_{s}}} Q\left(-\frac{\alpha_{s}}{\sqrt{2 \beta_{s}}}\right)$. The relation $Q(x)=1-Q(-x)$ proves helpful to rewrite $\mathbb{I}_{2}$ as $\frac{1}{2 \beta_{s}}+\sqrt{\frac{\pi}{\beta_{s}}} \frac{\alpha_{s}}{2\beta_{s}} e^{\frac{\alpha^{2}_{s}}{4 \beta_{s}}}\left\{1-Q\left(\frac{\alpha_{s}}{\sqrt{2 \beta_{s}}}\right)\right\}$. As a result, we can rewrite $\zeta(\alpha_{s},\beta_{s})$ as
\begin{equation}
\label{eq:62}
\begin{aligned}
\zeta(\alpha_{s},\beta_{s}) &=\mathbb{I}_{1}+\mathbb{I}_{2}\\
&= \frac{1}{\beta_{s}}+\sqrt{\frac{\pi}{\beta_{s}}} \frac{\alpha_{s}}{2\beta_{s}} e^{\frac{\alpha^{2}_{s}}{4 \beta_{s}}}-\sqrt{\frac{\pi}{\beta_{s}}} \frac{\alpha_{s}}{\beta_{s}} e^{\frac{\alpha^{2}_{s}}{4 \beta_{s}}} Q\left(\frac{\alpha_{s}}{\sqrt{2 \beta_{s}}}\right).
\end{aligned}
\end{equation}
Hence, we  obtain our ultimate formula for $R_{\mathbf{y}}(i,j)$ as below:
\begin{equation}
\label{eq:63}
\begin{aligned}
R_{\mathbf{y}}(i,j)=\frac{e^{\frac{-d^{2}}{p_{0}+p_{l}}}}{\pi\sqrt{\left(p_{0}^{2}-p_{l}^{2}\right)}}\left\{ \int_{0}^{\frac{\pi}{2}} \frac{1}{\beta_{s}}+\sqrt{\frac{\pi}{\beta_{s}}} \frac{\alpha_{s}}{2\beta_{s}} e^{\frac{\alpha^{2}_{s}}{4 \beta_{s}}}\right.\\\left.-\sqrt{\frac{\pi}{\beta_{s}}} \frac{\alpha_{s}}{\beta_{s}} Q\left(\frac{\alpha_{s}}{\sqrt{2 \beta_{s}}}\right) e^{\frac{a^{2}}{4 \beta_{s}}} \,d\theta\right\}-1.
\end{aligned}
\end{equation}

\section{Proof of The Convexity of $\Phi(p_{l})$ in (\ref{eq:91})}
Since $\log(\cdot)$ is a strictly increasing function, it is thus only required to analyze the criterion $\Phi_{m}(p_{l})=\left(R_{\mathbf{y}}(l)-J_{s}(p_{l})\right)^{2}$ to show the convexity of $\Phi(p_{l})$. Taking the derivative of $\Phi_{m}(p_{l})$ with respect to $p_{l}$ results in
\begin{equation}
\label{eq:app_1}
\Phi_{m}^{\prime}(p_{l}) = -2\left(R_{\mathbf{y}}(l)-J_{s}(p_{l})\right)J_{s}^{\prime}(p_{l}),
\end{equation}
where $J_{s}$ is the approximated version of (\ref{eq:112}) using the Gauss-Legendre quadrature presented in (\ref{eq:89}). Mathematically, $J_{s}(p_{l})$ can be represented by the following closed-form formula:
\begin{equation}
\label{eq:app_2}
\begin{aligned}
J_{s}(p_{l})=\frac{e^{-\frac{d^{2}}{p_{0}+p_{l}}}}{\pi}&\left(\pi+2\tan^{-1}\left(\frac{p_{l}}{\sqrt{p_{0}^{2}-p_{l}^{2}}}\right)\right.\\&\left.+\frac{\pi I}{4\sqrt{p_{0}^{2}-p_{l}^{2}}}\right)-1,
\end{aligned}
\end{equation}
where $I$ is given by
\begin{equation}
\label{eq:app_3}
I = \sum_{i=1}^{N_{q}} \omega_{i}\sqrt{\frac{\pi}{\beta_{s}}}\left(\frac{\alpha_{s}}{\beta_{s}}\right)\left(\frac{1}{2}-Q\left(\frac{\alpha_{s}}{\sqrt{2\beta_{s}}}\right)\right)e^{\frac{\alpha_{s}^{2}}{4\beta_{s}}}.
\end{equation}
Based on (\ref{eq:app_2}) and (\ref{eq:app_3}), $J_{s}^{\prime}(p_{l})$ can be written as
\begin{equation}
\label{eq:app_4}
\begin{aligned}
J_{s}^{\prime}(p_{l})&=e^{-\frac{d^{2}}{p_{0}+p_{l}}}\left(\frac{2}{\pi\sqrt{p_{0}^{2}-p_{l}^{2}}}+\frac{d^{2}\left(\pi+2\sin^{-1}\left(\frac{p_{l}}{p_{0}}\right)\right)}{\pi\left(p_{0}+p_{l}\right)^{2}}\right)\\&+\frac{e^{-\frac{d^{2}}{p_{0}+p_{l}}}}{4\sqrt{p_{0}^{2}-p_{l}^{2}}}\left(\frac{\partial I}{\partial \alpha_{s}}\frac{\partial \alpha_{s}}{\partial p_{l}}+\frac{\partial I}{\partial \beta_{s}}\frac{\partial \beta_{s}}{\partial p_{l}}\right)\\&+\left(\frac{e^{-\frac{d^{2}}{p_{0}+p_{l}}}\left(d^{2}\left(p_{0}-p_{l}\right)+p_{l}^{2}+p_{0}p_{l}\right)}{4\left(p_{0}+p_{l}\right)\left(p_{0}^{2}-p_{l}^{2}\right)^{3/2}}\right)I,
\end{aligned}
\end{equation}
where $\frac{\partial \alpha_{s}}{\partial p_{l}}$ and $\frac{\partial \beta_{s}}{\partial p_{l}}$ are given according to (\ref{eq:113}).
As can be seen in (\ref{eq:app_1}), (\ref{eq:app_2}) and (\ref{eq:app_4}), analyzing the convexity of $\Phi_{m}(p_{l})$ depends on the parameters $d$, $p_{0}$, $N_{q}$ and $\{\theta_{i}\}$ which indicates the fact that the analysis is restricted to the case where the mentioned parameters are known; i.e. the parameters must be specified for the covariance matrix recovery. Generally, based on (\ref{eq:app_1}), (\ref{eq:app_2}) and (\ref{eq:app_4}), $\Phi_{m}(p_{l})$ is convex when $J_{s}^{\prime}(p_{l})>0$ or equivalently $J_{s}(p_{l})$ is a strictly increasing function in the feasible region of $p_{l}$; i.e. $-p_{0} \leq p_{l} \leq p_{0}$. As a result, $\Phi_{m}^{\prime}(p_{l})=0$ has only one solution which is the value of $p_{l}$ that satisfies $R_{\mathbf{y}}(l)=J_{s}(p_{l})$. Therefore, the convexity of $\Phi_{m}(p_{l})$ can be easily concluded based on (\ref{eq:app_1}). For instance, one can easily verify that the selected parameters for the recovery of the input covariance matrix in Section~\ref{subsec:18} makes $J_{s}(p_{l})$ a strictly increasing function, and thus,  $\Phi_{m}(p_{l})$ a convex function.

\section{Proof of The Modified Bussgang Law Formula}
Note that
\begin{equation}
\label{eq:95}
\begin{aligned}
R_{\mathbf{yw}}(i,j) &= \frac{e^{-\frac{d^{2}}{p_0+p_{l}}}}{2\pi\sqrt{p_{0}^2-p_{l}^2}} \int^{\infty}_{-\infty}  g(w_{j})e^{\frac{2d(p_{0}-p_{l})w_{j}+w^{2}_{j}p_{0}}{-2(p^{2}_{0}-p^{2}_{l})}}\\ &\int^{\infty}_{-\infty}w_{i}e^{\frac{2d(p_{0}-p_{l})w_{i}+w^{2}_{i}p_{0}-2p_{ij}w_{i}w_{j}}{-2(p^{2}_{0}-p^{2}_{l})}} \,dw_{i}\,dw_{j}.
\end{aligned}
\end{equation}
Let us denote  the inner integral and the outer integral by $\mathbb{L}_{1}$ and $\mathbb{L}_{2}$, respectively. The inner integral is evaluated as,
\begin{equation}
\label{eq:96}
\begin{aligned}
\mathbb{L}_{1}&=\int^{\infty}_{-\infty}w_{i}e^{\frac{2d(p_{0}-p_{l})w_{i}+w^{2}_{i}p_{0}-2p_{l}w_{i}w_{j}}{-2(p^{2}_{0}-p^{2}_{l})}} \,dw_{i}\\ &=e^{\frac{\left(p_{0}d-p_{l}\left(w_{j}+d\right)\right)^2}{2p_{0}(p^{2}_{0}-p^{2}_{l})}}\int^{\infty}_{-\infty}
w_{i}e^{-\frac{\left(w_{i}+\left(d-\frac{p_{l}}{p_{0}}\left(w_{j}+d\right)\right)\right)^2}{2(p_{0}-\frac{p^{2}_{l}}{p_{0}})}}\,dw_{i}\\ &= e^{\frac{\left(p_{0}d-p_{l}\left(w_{j}+d\right)\right)^2}{2p_{0}(p^{2}_{0}-p^{2}_{l})}} \sqrt{2\pi \left((p_{0}-\frac{p^{2}_{l}}{p_{0}})\right)} \times \cdots\\
&\left(\frac{p_{l}}{p_{0}}\left(w_{j}+d\right)-d\right).
\end{aligned}
\end{equation}
Moreover, the outer integral may be evaluated as,
\begin{equation}
\label{eq:97}
\begin{aligned}
&\mathbb{L}_{2} = \sqrt{2\pi \left((p_{0}-\frac{p^{2}_{l}}{p_{0}})\right)} e^{\frac{p^{2}_{0}d^{2}+p^{2}_{l}d^{2}-2dp_{0}p_{l}}{2p_{0}(p^{2}_{0}-p^{2}_{l})}} \times \cdots\\
&\int^{\infty}_{-\infty}g(w_{j})\left(\frac{p_{l}}{p_{0}}\left(w_{j}+d\right)-d\right)e^{\frac{w^{2}_{j}(p^{2}_{0}-p^{2}_{l})+2dw_{j}(p^{2}_{0}-p^{2}_{l})}{-2p_{0}(p^{2}_{0}-p^{2}_{l})}} \,dw_{j}.
\end{aligned}
\end{equation}
The integration in (\ref{eq:97}) can thus be simplified as follows:
\begin{equation}
\label{eq:170}
\begin{aligned}
\mathbb{L}_{2} &= \sqrt{2\pi \left(p_{0}-\frac{p^{2}_{l}}{p_{0}}\right)} e^{\frac{d^{2}}{p_0+p_{l}}} \times \cdots\\
&\int^{\infty}_{-\infty}g(w_{j})\left(\frac{p_{l}}{p_{0}}\left(w_{j}+d\right)-d\right) e^{-\frac{\left(w_{j}+d\right)^2}{2p_{0}}} \,dw_{j}.
\end{aligned}
\end{equation}
Therefore, based on (\ref{eq:95}), (\ref{eq:96}) and (\ref{eq:170}), the  modified Bussgang law is obtained as
\begin{equation}
\label{eq:99}
\begin{aligned}
R_{\mathbf{yw}}(i,j) =  C_{1}p_{l}-C_{2} d(p_{0}-p_{l}),
\end{aligned}
\end{equation}
where $C_{1}$ and $C_{2}$ are given by
\begin{equation}
\label{eq:150}
\begin{aligned}
C_{1} &= \frac{1}{\sqrt{2\pi p^{3}_{0}}}\int^{\infty}_{-\infty} w_{j}g(w_{j})e^{-\frac{(w_{j}+d)^{2}}{2p_{0}}}\,dw_{j},\\ C_{2} &= \frac{1}{\sqrt{2\pi p^{3}_{0}}}\int^{\infty}_{-\infty}g(w_{j})e^{-\frac{(w_{j}+d)^{2}}{2p_{0}}}\,dw_{j}.
\end{aligned}
\end{equation}
\bibliographystyle{IEEEbib}
\bibliography{strings,refs}

\end{document}